\documentclass[11pt,onecolumn]{IEEEtran}
\usepackage[T1]{fontenc}
\usepackage[utf8]{inputenc}
\usepackage{lmodern}
\usepackage{amsmath,amssymb,amsthm,mathtools}
\usepackage{enumitem}
\usepackage{microtype}
\usepackage{cite}
\usepackage{xcolor}
\usepackage[hidelinks]{hyperref}
\newcommand{\al}[1]{{\color{black}{#1}}}
\newtheorem{theorem}{Theorem}[section]
\newtheorem{lemma}[theorem]{Lemma}
\newtheorem{proposition}[theorem]{Proposition}
\newtheorem{corollary}[theorem]{Corollary}
\theoremstyle{definition}
\newtheorem{definition}[theorem]{Definition}

\newtheorem{remark}[theorem]{Remark}

\newcommand{\PR}{\mathrm{PR}}
\newcommand{\FS}{\mathrm{FS}}
\newcommand{\comp}{\mathrm{comp}}
\newcommand{\MLR}{\operatorname{MLR}}
\newcommand{\Dim}{\operatorname{Dim}}

\newcommand{\E}{\mathbb E}
\newcommand{\ind}{\mathbf 1}
\newcommand{\bits}{\{0,1\}}
\newcommand{\orcid}[1]{\href{https://orcid.org/#1}{ {\includegraphics[scale=0.5]{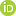}}}}

\title{Universal Individual-Sequence Prediction with a Primitive-Recursive Superpredictor}
\author{Amir Leshem\thanks{The author is with Faculty of Engineering, Bar-Ilan University, Ramat Gan 52900, Israel. Email: amir.leshem@biu.ac.il. This research was partially supported by grant ISF 2197/22.}\textsuperscript{\orcid{0000-0002-2265-7463}}, ~\IEEEmembership{Fellow,~IEEE}}
\date{}

\begin{document}
\maketitle

\begin{abstract}
\al{
We study universal prediction of individual binary sequences
under conditional expected zero--one loss, using
primitive-recursive (PR) forecasters as a computational
benchmark. Adapting learning-theoretic exponential weighting
and delayed expert activation, we construct a computable
PR-superpredictor with explicit sublinear regret against
every rational-valued PR forecaster on every individual
sequence. No primitive-recursive predictor can achieve
this universal property.

We prove that the Prediction by Partial Matching (PPM)
predictor is primitive recursive, implying that our
superpredictor attains the Bayes error on every
Martin-L\"of random realization of a computable stationary
ergodic source. We extend this optimality to nonstationary
sources, including primitive-recursive interleavings of
independent ergodic processes and finite-state hidden
Markov models with primitive-recursive rational
time-varying parameters. For the latter, the excess
cumulative loss relative to the conditional Bayes
benchmark is $O(\sqrt{N\log N})$, including for the error processes of nonstationary Gilbert--Elliott channels.

Finally, we construct a computable sequence on which every
primitive-recursive predictor has asymptotic error $1/2$ while a
PR-superpredictor with a suitable effective expert ordering has
vanishing normalized loss, and we derive strict separations
at every Grzegorczyk level, including finite-state prediction.
}
\end{abstract}

\begin{IEEEkeywords}
Universal prediction, individual sequences, prediction with expert advice,
primitive recursive functions, Kolmogorov complexity, PPM.
\end{IEEEkeywords}

\section{Introduction}
\label{sec:introduction}

Universal prediction belongs to the broader individual-sequence approach
initiated by Ziv and Lempel in their work on finite-sequence complexity and
universal compression \cite{LempelZiv1976,ZivLempel1978}.  In this approach,
the observed sequence is treated as a deterministic object, without assuming
that it was generated by a known probabilistic source, and performance is
measured relative to a prescribed class of admissible machines or decision
rules.  Merhav's recent overview \cite{MerhavZivIndividualSequence} places this
program within a hierarchy of progressively weaker modeling assumptions: from
known memoryless and parametric sources, through Markov, finite-state, and
stationary ergodic source classes, to the individual-sequence setting in which
no stochastic mechanism is postulated for the data.  Universal prediction asks
for a single sequential rule that performs asymptotically as well as the best
predictor in the chosen reference class, without knowing in advance which rule
is appropriate for the observed sequence.

\al{For finite-memory (tree) sources, Weinberger, Rissanen, and
Feder~\cite{WeinbergerRissanenFeder1995} constructed a sequential
universal source, based on the Context algorithm, that attains
Rissanen's lower bound on the redundancy. The context-tree weighting
method of Willems, Shtarkov, and
Tjalkens~\cite{WillemsShtarkovTjalkens1995} obtains redundancy
guarantees of the same order by Bayesian weighting over all tree
models and their parameters. It avoids explicit estimation of the
model and yields an efficient sequential algorithm. Weighting over
transition patterns extends the approach to memoryless sources whose
parameter changes at unknown times~\cite{Willems1996}.
Universal prediction under stationary ergodic assumptions
has a complementary history. Ornstein
\cite{ornstein1978guessing} established fundamental
results on predicting stationary processes, while
Algoet \cite{algoet1992universal} developed universal
schemes for prediction, gambling, and portfolio
selection. Ryabko \cite{ryabko1984twice,ryabko1988prediction}
established important connections between universal coding and prediction, including universal mixtures
of finite-order Markov models that preceded the
PPM measure used here. These results establish
universality relative to stochastic source classes,
whereas our primary objective is pathwise competition
with a computationally defined predictor class on
every individual binary sequence. D.~Ryabko later showed that, under total variation or
expected average KL divergence, whenever a class of
stochastic processes admits a universal predictor,
a countably supported Bayesian mixture also suffices
\cite{Ryabko2010}.
}

A parallel hierarchy concerns the complexity assigned to an individual
sequence.  Finite-state complexity is operational and computable, but it can
regard an algorithmically simple sequence as maximally complex.  A canonical
example is the counting sequence obtained by concatenating all binary words:
it has a short programmatic description and therefore vanishing normalized
Kolmogorov complexity, while its finite-state complexity is maximal.  By
contrast, Kolmogorov complexity captures unrestricted algorithmic structure
but is itself uncomputable.  This contrast suggests that between finite-state
mechanisms and unrestricted computation there may be useful intermediate
classes that are substantially more expressive than finite-state rules while
retaining enough effective structure to support universal procedures.  The
same question arises in prediction: one seeks a benchmark rich enough to
recognize algorithmic regularity, but sufficiently structured that a single
computable master can compete with all of its members.
Within the individual-sequence framework, Feder, Merhav, and Gutman developed
the fundamental theory of sequential prediction relative to finite-state
predictors and constructed randomized schemes attaining finite-state
predictability \cite{FederMerhavGutman1992}.  Merhav and Feder later surveyed
and extended this information-theoretic framework \cite{MerhavFeder1998}.
Ziv and Merhav considered a different enlargement based on context-tree
classes whose number of states may grow with the horizon
\cite{ZivMerhav2007}.  Their benchmark is horizon dependent: the minimizing
context tree may change with the observed prefix length.  Our benchmark is of
a different type.  We compete with every fixed prediction rule belonging to a
broad computational class, while requiring a guarantee on every individual
binary sequence.
\al{
Universal forecasting relative to computationally
restricted classes was also studied by Vovk
\cite{vovk1992universal}, who constructed forecasting
algorithms competitive with simple forecasting
systems under explicit computational time bounds,
for a fixed horizon and time budget.
The general methodology of exponential weighting
and aggregation of expert advice is well established
\cite{vovk1990aggregating,vovk1998game,CesaBianchiLugosi2006}.
Our contribution is applying an effective aggregation to the entire class of primitive-recursive forecasters, whose totality is guaranteed
syntactically without imposing a common
computational resource bound. We give an explicit total computable aggregate, establishing the
separation between the aggregate and its reference class.
}

\al{
The broadest natural computational benchmark is the class of all
total computable predictors. However, this class does not admit
an effective enumeration consisting of exactly its members, and
a diagonal argument shows that no computable probabilistic
predictor can compete with every total computable predictor
on every individual sequence.
\al{This diagonal obstruction is closely related to
Putnam's argument against universal inductive methods;
see Sterkenburg \cite{sterkenburg2019putnam}}
\al{Solomonoff's universal mixture over all computable
measures~\cite{Solomonoff1964a,Solomonoff1964b,Solomonoff1978} predicts
every computable source, but it is only lower semicomputable, not computable; Putnam's argument shows that this loss of computability is
unavoidable for a method that is universal over all computable
hypotheses.}
Aggregation works for any effectively enumerable
class of total predictors with a uniform computable evaluator. Such
classes range from finite-state and explicitly time- or space-bounded
predictors, through primitive-recursive ones, to predictors whose
totality is provable in Peano arithmetic or another fixed, sound,
effectively axiomatized theory. None contains all total computable
predictors, and each guarantee is relative to the chosen class.
Primitive recursion is far richer than finite-state prediction: a
primitive-recursive predictor may scan the entire observed prefix,
perform bounded searches and arithmetic, and use memory that grows
according to any primitive-recursive schedule.}

The deterministic predictor classes satisfy \al{(Corollary~\ref{cor:strict-inclusions})}
\[
    \mathcal P_{\FS}
    \subsetneq
    \mathcal P_{\PR}
    \subsetneq
    \mathcal P_{\comp}.
\]
For aggregation, we use the larger syntactically enumerable class
$\mathcal Q_{\mathrm{PR}}^{\mathrm{rat}}$ of rational-valued
primitive-recursive forecasters.  It contains
$\mathcal P_{\mathrm{PR}}$ as the subclass of $\{0,1\}$-valued experts.
We call a computable probabilistic master with sublinear regret relative to
every member of $\mathcal Q_{\mathrm{PR}}^{\mathrm{rat}}$ a
\emph{PR-superpredictor}.  The prefix ``PR'' describes the expert class, not
the complexity of the master.  In the represented-real terminology introduced
in Section~\ref{sec:computable}, the predictor constructed below satisfies
\[
Q_{\PR}
\in
\mathcal Q_{\mathrm{comp}}^{\mathrm{real}}
\setminus
\mathcal Q_{\mathrm{PR}}^{\mathrm{real}}.
\]
Primitive recursion is not a conventional feasible complexity
class: its running times may be astronomically large.
\al{
Nevertheless, its use as a computational benchmark leads to two principal contributions.

First, we establish Bayes optimality for computable
stationary ergodic and structured nonstationary sources.
Steifer~\cite{Steifer2019PhD} established an effective
Bayes lower bound for computable deterministic predictors
and constructed a universal stationary-ergodic predictor.
Building on these results and the PPM universality theorem
of D\k{e}bowski and Steifer~\cite{DebowskiSteifer2022},
our contribution is to prove that PPM is primitive recursive,
extend the Bayes lower bound to computable probabilistic
predictors, and obtain pathwise universality through a
single computable PR-superpredictor. We further establish
Bayes optimality for primitive-recursive interleavings of
independent ergodic sources and nonstationary hidden Markov
models with primitive-recursive parameters.

Second, we establish strict computational separations
between predictor classes. A primitive-recursive normal
sequence defeats every finite-state predictor while admitting
perfect PR prediction. More strongly, we construct a
computable sequence on which every PR forecaster has
asymptotic loss $1/2$, whereas a suitably ordered
PR-superpredictor has vanishing loss. These separations,
which go beyond the stationary-ergodic prediction framework
of~\cite{Steifer2019PhD}, demonstrate the advantage of
effective aggregation over individual computationally
restricted predictors.}

Throughout the paper, a probabilistic predictor
$Q:2^{<\omega}\to[0,1]$ assigns a probability to the next bit being $1$.  We
write
\begin{align}
    L_N(Q,X)
      :=&\sum_{n<N}\ell\!\left(Q(X_0^{n-1}),X_n\right),
    \qquad \\
    \nonumber
    \ell(q,b):=&q(1-b)+(1-q)b,
\end{align}
for its cumulative conditional expected zero-one loss.  A deterministic
predictor is identified with a $\{0,1\}$-valued probabilistic predictor, so the
same notation $L_N(P,X)$ is its cumulative number of errors.  Average loss is
therefore written explicitly as $L_N/N$. Throughout, $\log$ denotes the natural logarithm, while $\log_2$ denotes the base-two logarithm.

The contributions of the paper are as follows.
\begin{enumerate}[leftmargin=*,label=(\roman*)]
\item \al{We relate prediction-error density to algorithmic complexity:}
      for a total computable predictor, an upper prediction-error rate
      $\gamma\leq1/2$ implies the effective packing-dimension bound
      $\Dim(X)\leq H_2(\gamma)$.  We also prove that no computable
      probabilistic predictor can compete on every individual sequence with
      all total computable deterministic predictors.
      \item We introduce primitive-recursive predictors as a broad syntactically
      enumerable class of total experts and construct a computable
      PR-superpredictor.  Expert $e$ is activated at a time $\tau_e=e(e+1)/2$, receives
      polynomial prior mass $\pi_e=1/((e+1)(e+2))$, and retains its accumulated
      exponential performance weight when the learning rate changes.  The
      resulting regret against every fixed PR expert is $o(N)$, with an
      explicit finite-time bound.  Conversely, no PR-computable probabilistic
      predictor can possess this superprediction property.
\item We prove that the deterministic PPM predictor is primitive recursive.
      Therefore, for every computable stationary ergodic binary measure $\mu$
      and every $X\in\MLR_\mu$, primitive-recursive prediction, unrestricted
      computable prediction, and the PR-superpredictor all attain the
      infinite-past Bayes error $e_\mu$.
\item We extend this optimality beyond stationarity.  If finitely many
      independent computable stationary ergodic sources are interleaved by an
      arbitrary primitive-recursive schedule, then, without being given either
      the schedule or the component laws, the PR-superpredictor satisfies
      \[
          L_N(Q_{\PR},X)
          =\sum_j N_j(N)e_j+o(N).
      \]
      No limiting state frequencies are required, and no computable predictor
      can asymptotically improve on this cumulative statewise Bayes benchmark.

\al{We further establish universal prediction for nonhomogeneous
hidden Markov sources with primitive-recursively varying rational
transition and emission probabilities. After $N$ prediction rounds,
the PR-superpredictor achieves the cumulative conditional Bayes
benchmark within $O(\sqrt{N\log N})$ on every Martin-L\"of random
realization. This includes the binary error process of
nonstationary Gilbert--Elliott channels
\cite{Gilbert1960,Elliott1963}. The one-state case of Theorem~\ref{thm:pr-nonstationary-hmm}
is a primitive-recursive signal in Bernoulli noise.}
\item \al{We prove a separation from the full primitive-recursive class:
      there is a computable sequence on which every primitive-recursive
      forecaster has asymptotic loss $1/2$, whereas a PR-superpredictor with a suitable effective
      expert ordering has loss $O(\sqrt N\log N)$. A level-by-level version inside the
      Grzegorczyk hierarchy follows from Di~Paola's subrecursive random
      sets. Its base case is a finite-state separation: there is a
      primitive-recursive normal sequence on which every finite-state
      predictor, deterministic or randomized, has error $1/2$, while
      primitive-recursive prediction has error $0$. These results yield
      the strict inclusions
      $\mathcal P_{\mathrm{FS}}\subsetneq\mathcal P_{\mathrm{PR}}
      \subsetneq\mathcal P_{\mathrm{comp}}$.}
\end{enumerate}

\al{The finite-state separation illustrates a classical phenomenon in
individual-sequence theory: a sequence can be algorithmically simple
yet unpredictable by finite-state machines. Our hierarchy theorem
extends this distinction to every fixed Grzegorczyk level, exhibiting
primitive-recursive sequences that are unpredictable at that level.
The full-PR separation goes further: a single computable sequence
defeats every PR forecaster, while a suitably ordered computable
PR-superpredictor achieves vanishing loss. Thus, computational
complexity limits individual predictors, while effective aggregation
can overcome those limits.}

The paper is related to, but distinct from, a companion study of prediction
innovations \cite{leshem2026sequential}.  The companion work treats the entire
residual sequence and asks which algorithmic-complexity and Turing-degree
structures can be exposed by prediction and computable sampling.  Here we
retain only the full-schedule density of prediction errors and optimize that
density over a predictor class.  Structural innovation spectra and optimal
Hamming prediction answer complementary questions.

The remainder of the paper is organized as follows.  Section~\ref{sec:computable}
introduces computable prediction, defines the primitive-recursive reference,
and develops the Kolmogorov-complexity bounds.  Section~\ref{sec:stochastic-examples}
develops Bernoulli, Markov, and general stationary-ergodic benchmarks.
Section~\ref{sec:pr-aggregation} constructs the PR-superpredictor.  Section~\ref{sec:ppm-ergodic}
proves stationary-ergodic optimality, Section~\ref{sec:pr-switching} proves
optimality for primitive-recursive nonstationary interleavings, hidden Markov sources, and Bayes optimality for noisy primitive-recursive signals.
Section~\ref{sec:separation} proves separation at every Grzegorczyk level with finite-state prediction as its base case, and the separation of
the full primitive-recursive class from its aggregate.  

\section{Computable Prediction and Algorithmic Complexity}
\label{sec:computable}

We work with one-sided binary sequences $X=X_0X_1\cdots\in\bits^{\mathbb N}$.
For $n\geq0$, write $X_0^{n-1}$ for the length-$n$ prefix, with
$X_0^{-1}$ equal to the empty string.

\begin{definition}[Predictors and losses]
A probabilistic predictor is a total map
\al{$Q:\bits^{<\mathbb N}\to[0,1]$},
where $Q(\sigma)$ is interpreted as the probability assigned to the next
symbol being $1$.  Its one-step conditional expected zero-one loss is
\al{$\ell(q,b):=q(1-b)+(1-q)b$},
and its cumulative loss on $X$ is
\[
    L_N(Q,X):=\sum_{n<N}\ell\!\left(Q(X_0^{n-1}),X_n\right).
\]
A deterministic predictor is a predictor
$P:\bits^{<\mathbb N}\to\bits$.  Identifying its output bit with a degenerate
probability, the same notation gives
\[
    L_N(P,X)
      =\sum_{n<N}\ind\!\left\{P(X_0^{n-1})\neq X_n\right\}.
\]
Thus $L_N$ always denotes cumulative loss; the corresponding average loss is $L_N/N$.
\end{definition}
\al{
\begin{remark}
Different prediction frameworks use different output spaces.
D\k{e}bowski and Steifer~\cite{DebowskiSteifer2022, Steifer2019PhD} define a predictor as a
deterministic next-symbol rule \al{$f:\mathcal X^*\to\mathcal X$},
while Morvai et al. \cite{moravi1996nonparametric} consider real-valued sequential predictions.
Here we use probabilistic predictors
\al{$Q:2^{<\omega}\to[0,1]$}, with deterministic predictors identified with the special case
$Q(\sigma)\in\{0,1\}$.

The distinction between real and rational outputs becomes important only when
computability restrictions are imposed.  A real-valued predictor requires an
effective representation of each output, whereas a rational-valued predictor
can be evaluated exactly by a finite computation.  Accordingly, the
primitive-recursive expert class introduced below is restricted to
rational-valued forecasters.  This restriction is made for effective
enumerability and exact evaluation of the experts; it is not a restriction on
the output of the aggregate, which is a computable real-valued predictor.
\end{remark}}
We now extend the definition to computable probabilistic predictors.
\begin{definition}[Computable probabilistic predictors]
A probabilistic predictor
$Q$
is \emph{computable} if there is a total computable rational-valued
function
\[
    \widehat Q:2^{<\omega}\times\mathbb N\to\mathbb Q
\]
such that
\[
    \left|\widehat Q(\sigma,k)-Q(\sigma)\right|
    \leq2^{-k}
\]
for every finite string $\sigma$ and every $k$.

It is \emph{PR-computable} if such an approximation
$\widehat Q(\sigma,k)$ can be chosen primitive recursive, using a fixed
primitive-recursive coding of rational numbers.  We write
$\mathcal Q_{\mathrm{comp}}^{\mathrm{real}}$
for the class of computable probabilistic predictors and
$\mathcal Q_{\mathrm{PR}}^{\mathrm{real}}$
for the class of PR-computable probabilistic predictors.
\end{definition}

\begin{remark}
\al{The aggregation method is not specific to primitive recursion:
it applies to any effectively enumerable class of total predictors
with uniformly computable outputs, and closure under composition or
primitive recursion is not required.}
\end{remark}
Thus computability of a real-valued predictor refers to uniform effective
approximation of its outputs, rather than to their being rational numbers. In particular, a computable predictor may take irrational values.

\subsection{The primitive-recursive reference class}
\label{subsec:pr-benchmark}

Fix a primitive-recursive coding of finite binary strings, with
primitive-recursive length, concatenation, and coordinate-extraction
operations.

\begin{definition}[Primitive-recursive predictor]
A deterministic predictor \al{$P:2^{<\omega}\to\{0,1\}$}
is \emph{primitive recursive} (PR) if its action on the
code of the observed finite prefix is primitive recursive.
We denote the class of such predictors by
\[
\mathcal P_{\mathrm{PR}}
:=\{P:P\text{ is a primitive-recursive predictor}\}.
\]
\end{definition}

For the aggregation theorem we now select a syntactically enumerable subclass
for which every expert has an exactly computable finite representation.
We therefore use rational-valued primitive-recursive forecasters.  The
rational restriction is a representation choice: it permits a uniform
enumeration and exact evaluation of the expert advice, while retaining the
deterministic primitive-recursive predictors as a subclass.

\begin{definition}[Rational primitive-recursive forecaster]
Let $a,b:2^{<\omega}\to\mathbb N$ be primitive-recursive functions.  Define
\[
    R_{a,b}(\sigma)
    :=
    \begin{cases}
    \displaystyle
    \frac{a(\sigma)}{a(\sigma)+b(\sigma)},
       &a(\sigma)+b(\sigma)>0,\\[3mm]
    \displaystyle\frac12,
       &a(\sigma)+b(\sigma)=0.
    \end{cases}
\]
A probabilistic predictor of this form is called a
\emph{rational primitive-recursive forecaster}.  We denote their class by
$
    \mathcal Q_{\mathrm{PR}}^{\mathrm{rat}}.
$
\end{definition}
Every deterministic primitive-recursive predictor belongs to
$\mathcal Q_{\mathrm{PR}}^{\mathrm{rat}}$.  Indeed, for
$P\in\mathcal P_{\mathrm{PR}}$, take
\al{$a(\sigma)=P(\sigma)$ and $b(\sigma)=1-P(\sigma)$.}
Thus
\[
    \mathcal P_{\mathrm{PR}}
    \subseteq
    \mathcal Q_{\mathrm{PR}}^{\mathrm{rat}}
    \subseteq
    \mathcal Q^{\mathrm{real}}_{\mathrm{PR}}.
\]


\al{Primitive-recursive programs form an effectively
enumerable class of total functions, with a total
computable uniform evaluator
\cite{Rogers1967,Odifreddi1989}. Consequently,
$\mathcal Q_{\mathrm{PR}}^{\mathrm{rat}}$ admits an
effective enumeration $P_0,P_1,\ldots$, possibly
with repetitions, whose rational outputs can be
evaluated uniformly. By the standard diagonal
argument, the uniform evaluator cannot itself
be primitive recursive.}

For a deterministic predictor, define the innovation sequence
\al{$I_P^X(n)=X_n\oplus P(X_0^{n-1})$.} Thus $I_P^X(n)=1$ exactly at the prediction errors, and
$L_N(P,X)=\sum_{n<N}I_P^X(n)$.

\begin{definition}[Computable prediction error]
Let
\[
    e_{\comp}(X)
      =\inf_{P\text{ total computable}}
        \limsup_{N\to\infty}\frac{L_N(P,X)}{N}.
\]
\al{For a class $\mathcal C$ of deterministic predictors,
$e_{\mathcal C}(X)$ is defined in the same way with the infimum
over $P\in\mathcal C$; in particular, $e_{\mathrm{PR}}(X)$ and
$e_{\mathrm{FS}}(X)$ are the benchmarks of $\mathcal P_{\mathrm{PR}}$
and $\mathcal P_{\mathrm{FS}}$.}
We also write
\[
    \Dim(X)=\limsup_{N\to\infty}
       \frac{K(X_0^{N-1})}{N}
\]
for the effective packing, or effective strong, dimension of $X$, where $K$
is prefix-free Kolmogorov complexity \cite{AthreyaHitchcockLutzMayordomo2007,
LiVitanyi2008}.
\end{definition}

\al{The normalized loss $L_N(P,X)/N$ coincides with
the finite-prefix prediction error defined by
Steifer~\cite{Steifer2019PhD}.
Our benchmark $e_{\mathrm{comp}}(X)$ takes the infimum
of the asymptotic upper prediction-error rates over
all total computable deterministic predictors.}
\al{The structural properties of innovation sequences are
studied in~\cite{leshem2026sequential}. In particular,
finite innovation is equivalent to computability
\cite{leshem2026sequential}, while the full
innovation transform is a computable measure-preserving
homeomorphism of Cantor space
\cite{leshem2026sequential}.
The present paper is concerned instead with the
asymptotic density of prediction errors and its
optimization over computationally restricted
predictor classes.}

\al{For the present paper, the relevant quantitative
observation is that the positions of prediction errors
can be encoded efficiently. Given a total computable
predictor $P$, the prefix $X_0^{N-1}$ can be
reconstructed from $P$ and the set of error positions.
Consequently, if}
\[
\al{
\limsup_{N\to\infty}\frac{L_N(P,X)}{N}
\leq \gamma\leq\frac12,
}
\]
\al{then}
\[
\al{\Dim(X)\leq H_2(\gamma),}
\]
\al{where $H_2$ is the binary entropy function.
This follows by encoding the error positions using
at most $NH_2(\gamma)+o(N)$ bits. Related connections between prediction and
resource-bounded effective dimension were established
by Fortnow and Lutz~\cite{fortnow2005prediction}.}

\begin{proposition}[Error-rate coding bound]
\label{prop:error-coding}
Let $P$ be a total computable predictor.  If
\[
    \limsup_{N\to\infty}\frac{L_N(P,X)}{N}\leq \gamma,
    \qquad 0\leq \gamma\leq\frac12,
\]
then
\[
    \Dim(X)\leq H_2(\gamma),
\]
where
\[
    H_2(t)=-t\log_2t-(1-t)\log_2(1-t).
\]
In particular,
\[
    \Dim(X)\leq H_2(e_{\comp}(X))
\]
whenever $e_{\comp}(X)\leq1/2$.
\end{proposition}

\begin{proof}
Given $P$, $N$, and the error set
$E_N=\{n<N:I_P^X(n)=1\}$, one reconstructs $X_0^{N-1}$ recursively.  If
$|E_N|\leq(\gamma+\delta)N$, then $E_N$ can be specified by its rank among all
subsets of $\{0,\ldots,N-1\}$ of size at most $(\gamma+\delta)N$.  Standard
binomial estimates give a description of length
\al{$NH_2(\min\{\gamma+\delta,1/2\})+o(N)$.} Taking the limsup and then $\delta\downarrow0$ proves the first assertion.
For the second, choose for every $\delta>0$ a computable predictor whose upper
error rate is at most $e_{\comp}(X)+\delta$ and use continuity of $H_2$.
\end{proof}

\al{Proposition~\ref{prop:error-coding} relates prediction-error
density to algorithmic complexity, while the structure of
innovation sequences and their computably selected traces
is studied in~\cite{leshem2026sequential}.}

\begin{proposition}[No uniformly best computable predictor]
\label{prop:no-universal-computable-master}
For every computable probabilistic predictor
\[
    Q:\bits^{<\mathbb N}\to[0,1],
\]
there exist a computable sequence $X$ and a total computable deterministic
predictor $P_X$ such that
\[
    L_N(P_X,X)=0
    \qquad\text{for all }N,
\]
whereas
\[
    L_N(Q,X)\geq \frac{N}{3}
    \qquad\text{for all }N.
\]
Consequently, there is no computable probabilistic predictor that is
uniformly asymptotically competitive with every total computable
deterministic predictor on every individual sequence.
\end{proposition}

\begin{proof}
Construct $X$ recursively.  Having computed $X_0^{n-1}$, compute a rational
number $\widetilde q_n$ satisfying
\[
    \bigl|\widetilde q_n-Q(X_0^{n-1})\bigr|<\frac16.
\]
Set
\[
    X_n=
    \begin{cases}
       0,&\widetilde q_n\geq1/2,\\
       1,&\widetilde q_n<1/2.
    \end{cases}
\]
If $X_n=0$, then $Q(X_0^{n-1})>1/3$; if $X_n=1$, then
$1-Q(X_0^{n-1})>1/3$.  Thus the conditional expected loss of $Q$ is at
least $1/3$ at every round.  The recursive construction makes $X$
computable.  Therefore \al{$P_X(\sigma):=X_{|\sigma|}$}
is a total computable predictor and is correct on $X$ at every round.
\end{proof}

Proposition~\ref{prop:no-universal-computable-master} complements the
non-enumerability of total computable programs.  The problem is not merely
that the natural expert pool is difficult to list: unrestricted computable
prediction admits no computable universally competitive master.  Primitive
recursion provides a large total class for which effective aggregation is
possible.

\section{Benchmarks from Simple Stochastic Sources}
\label{sec:stochastic-examples}

The individual-sequence benchmark becomes concrete on sample paths of
computable stochastic sources. \al{We use Martin-L\"of randomness to state effective pointwise versions of the
particular strong-law, ergodic, and martingale results invoked
below~\cite{Nies2009,BienvenuDayHoyrupMezhirovShen2012,DebowskiSteifer2022}}.

\subsection{Computable measures and Martin-L\"of randomness}
\label{subsec:mlr}

For a finite binary string $\sigma$, let
\[
    [\sigma]:=\{X\in2^\omega:\sigma\prec X\}
\]
denote the corresponding cylinder set.  A probability measure $\mu$ on
$2^\omega$ is \emph{computable} if, uniformly in $\sigma$ and $m$, one can
compute a rational approximation to $\mu([\sigma])$ with error at most
$2^{-m}$.  A sequence of open sets $(U_k)_{k\geq1}$ is \emph{uniformly
effectively open} if there is a computably enumerable set
$W\subseteq\mathbb N\times2^{<\omega}$ such that
\[
    U_k=\bigcup_{(k,\sigma)\in W}[\sigma].
\]

\begin{definition}[Martin-L\"of randomness]
Let $\mu$ be a computable probability measure on $2^\omega$.  A
\emph{$\mu$-Martin-L\"of test} is a uniformly effectively open sequence
$(U_k)_{k\geq1}$ satisfying $\mu(U_k)\leq2^{-k}$ for every $k$.  A sequence
$X\in2^\omega$ is \emph{Martin-L\"of random with respect to $\mu$} if
\[
    X\notin\bigcap_{k\geq1}U_k
\]
for every $\mu$-Martin-L\"of test.  The class of such sequences is denoted by
$\MLR_\mu$.  Relative to an oracle $A$, the cylinder enumeration is allowed to
be computably enumerable in $A$; the corresponding class is denoted by
$\MLR_\mu^A$.
\end{definition}

For stationary sources, the infinite-past conditional probability used below
is evaluated in the natural two-sided stationary extension.  We invoke only
specific effective results: The effective strong law, effective Birkhoff and Breiman ergodic theorems, effective martingale bounds, randomness conservation,
and the relative product-randomness theorem with the hypotheses stated at
the points where they are used
\cite{BienvenuDayHoyrupMezhirovShen2012,Nies2009,DebowskiSteifer2022}.

\subsection{Bernoulli sources}

Let $\mu_p$ be the Bernoulli measure with computable parameter
$p\in[0,1]$.

\al{The primitive-recursive empirical-majority predictor}
\begin{equation}
    P_{\mathrm B}(\sigma)
=
\mathbf{1}\!\left\{\|\sigma\|_1>\frac{|\sigma|}{2}\right\}
\end{equation}
\al{attains the Bayes error without knowing $p$. Indeed, for
every computable $p$ and every $X\in\MLR_{\mu_p}$,}
\[
    \lim_{N\to\infty}\frac{L_N(P_{\mathrm B},X)}{N}
      =\min(p,1-p).
\]
\al{For $p\neq1/2$, this follows from the effective strong law;
for $p=1/2$, from the effective martingale strong law.
The effective Bayes lower bound
\cite{Steifer2019PhD} and
Lemma~\ref{lem:probabilistic-bayes-lower-bound} then gives}
\[
    \al{e_{\mathrm{PR}}(X)=e_{\comp}(X)=\min(p,1-p).}
\]

For Bernoulli random points the complexity bound is tight:
\[
    \Dim(X)=H_2(p)=H_2(\min(p,1-p)).
\]

\subsection{Binary first-order Markov sources}

\al{Consider a stationary binary Markov source with computable
transition matrix}
\[
\al{
    M=\begin{pmatrix}1-a&a\\ b&1-b\end{pmatrix},
    \qquad a,b>0.
}
\]
\al{Here $a$ and $b$ are the probabilities of switching from
$0$ to $1$ and from $1$ to $0$, respectively. The stationary
probabilities are}
\[
    \pi_0=\frac{b}{a+b},\qquad
    \pi_1=\frac{a}{a+b}.
\]
\al{Since the optimal prediction depends only on the current
state, the Bayes error is}
\[
\al{
    e_M=\pi_0\min(a,1-a)+\pi_1\min(b,1-b).
}
\]

\al{Even when $a$ and $b$ are unknown, a primitive-recursive
predictor can estimate them from empirical transition
frequencies and attain $e_M$ on every
$\mu_M$-Martin--L\"of random sequence. This follows from
the effective ergodic theorem and the effective Bayes
lower bound; see~\cite{Steifer2019PhD}
and Lemma~\ref{lem:probabilistic-bayes-lower-bound}.}

\al{The entropy rate, in contrast, is}
\[
\al{
    h_M=\pi_0H_2(a)+\pi_1H_2(b),
}
\]
\al{which does not uniquely determine the Bayes error.
This distinction persists for general stationary
ergodic sources.}
\subsection{The stationary-ergodic Bayes benchmark}
\label{subsec:stationary_bayes}

\al{Let $\mu$ be a computable stationary ergodic binary
measure and $\bar\mu$ its natural two-sided extension.
Define the infinite-past conditional probability}
\[
    p_\infty(x)
    :=
    \bar\mu\!\left(
        X_0=1
        \,\middle|\,
        X_{-\infty}^{-1}=x_{-\infty}^{-1}
    \right)
\]
\al{and the corresponding Bayes error}
\[
    e_\mu
    :=
    \int
        \min\{p_\infty(x),1-p_\infty(x)\}
        \,d\bar\mu(x).
\]

\al{The two-sided extension defines the asymptotic
benchmark; online predictors observe only the finite
past $X_0^{n-1}$. The entropy rate is}
\[
    h_\mu
    =
    \int H_2(p_\infty(x))\,d\bar\mu(x),
\]
\al{and the binary entropy--prediction inequalities give}
\[
    H_2^{-1}(h_\mu)\leq e_\mu\leq\frac{h_\mu}{2},
\]
\al{where $H_2^{-1}$ is restricted to $[0,1/2]$;
see~\cite{DebowskiSteifer2022}.
Thus entropy rate and optimal Hamming prediction error
are distinct characteristics of a source.}

\al{D\k{e}bowski and Steifer established the effective
Bayes lower bound for computable deterministic predictors
and the universality of their PPM predictor on
Martin--L\"of random points
\cite{DebowskiSteifer2022}.
Section~\ref{sec:ppm-ergodic} proves that this PPM
predictor is primitive recursive and hence belongs
to the reference class of our superpredictor.}

\section{Universal Aggregation over Primitive-Recursive Experts}
\label{sec:pr-aggregation}

We now aggregate the primitive-recursive reference predictors introduced in
Section~\ref{subsec:pr-benchmark}.  Prediction rules may use unbounded but
algorithmically structured memory, while each expert is required to be
primitive recursive.  The aggregation argument is a countable-expert version
of exponential weighting with delayed activation.  Prediction with expert advice is treated
systematically in \cite{CesaBianchiLugosi2006,vovk1998game}, and countable expert
classes are treated explicitly in
\cite{ChernovKalnishkanZhdanovVovk2008}.  We give a direct construction because
effective enumeration, deferred activation, and preservation of earlier
performance weights are central here.

\subsection{PR-superprediction and the internal diagonal obstruction}

\begin{definition}[PR-superpredictor]
A total computable probabilistic predictor
\[
    Q:2^{<\omega}\to[0,1]
\]
is called a \emph{primitive-recursive superpredictor}
(\emph{PR-superpredictor}) if, for every
$R\in\mathcal Q_{\mathrm{PR}}^{\mathrm{rat}}$ and every $X\in2^\omega$,
\[
    \limsup_{N\to\infty}
    \frac{L_N(Q,X)-L_N(R,X)}{N}\leq0.
\]
Since
$\mathcal P_{\mathrm{PR}}
 \subseteq\mathcal Q_{\mathrm{PR}}^{\mathrm{rat}}$,
a PR-superpredictor is, in particular, asymptotically competitive with every
deterministic primitive-recursive predictor.
\end{definition}

\begin{lemma}[A PR-superpredictor cannot be primitive recursive]
\label{lem:superpredictor-not-pr}
No PR-computable probabilistic predictor is a PR-superpredictor.  In
particular, any predictor satisfying Theorem~\ref{thm:pr-aggregation} lies
strictly outside the primitive-recursive class, although it may be total
computable.
\end{lemma}

\begin{proof}
Let $Q$ be PR-computable, and construct a binary sequence $X$ recursively.
After $X_0^{n-1}$ has been defined, compute
\[
    \widehat q_n:=\widehat Q(X_0^{n-1},2),
\]
so that
$|\widehat q_n-Q(X_0^{n-1})|\leq1/4$, and set
\[
    X_n=
    \begin{cases}
       0,&\widehat q_n\geq1/2,\\
       1,&\widehat q_n<1/2.
    \end{cases}
\]
More explicitly, starting from the code of the empty string, the code of
each successive prefix is obtained by primitive recursion using the displayed
rule and the primitive-recursive concatenation operation.  Hence the resulting
sequence $n\mapsto X_n$ is primitive recursive.  If $X_n=0$, then
$Q(X_0^{n-1})\geq1/4$; if $X_n=1$, then
$1-Q(X_0^{n-1})>1/4$.  Hence
\[
    L_N(Q,X)\geq\frac{N}{4}
    \qquad\text{for every }N.
\]
On the other hand,
\[
    P_X(\sigma):=X_{|\sigma|}
\]
is a primitive-recursive predictor and satisfies $L_N(P_X,X)=0$ for all
$N$.  Therefore $Q$ has linear regret relative to the PR expert $P_X$ and
cannot be a PR-superpredictor.
\end{proof}

\subsection{A computable PR-superpredictor for the countable class}
Fix the effective enumeration
\[
    P_0,P_1,\ldots
\]
of $\mathcal Q_{\mathrm{PR}}^{\mathrm{rat}}$ introduced in
Section~\ref{subsec:pr-benchmark}.

\al{We distinguish the dyadic learning-rate epochs
from the expert-activation schedule.}
\[
    \al{b_r:=2^r-1,\qquad
    I_r:=\{b_r,\ldots,b_{r+1}-1\}.}
\]
The learning rate remains
\[
    \eta_r:=2^{-r/2},\qquad \eta_n:=\eta_r
    \quad(n\in I_r).
\]

\al{Independently, expert $P_e$ is activated at}
\[
    \al{\tau_e:=\frac{e(e+1)}2.}
\]
\al{At round $n$, the active set is
$A_n:=\{e:\tau_e\le n\}$, with
$|A_n|=\Theta(\sqrt{n+1})$.}

Throughout this construction, $P_e(\sigma)$ denotes the probability that
expert $e$ assigns to the next symbol being $1$.  Before activation,
expert $P_e$ is treated as sleeping and is assigned the aggregate's loss.
Accordingly, define
\[
    \widetilde\ell_{e,n}
    :=
    \begin{cases}
       \ell(q_n,X_n),&n<\tau_e,\\[1mm]
       \ell\bigl(P_e(X_0^{n-1}),X_n\bigr),&n\geq\tau_e,
    \end{cases}
\]
and
\[
    \widetilde L_{e,n}
    :=
    \sum_{t<n}\widetilde\ell_{e,t}.
\]

We use the polynomial prior
\[
    \pi_e:=\frac{1}{(e+1)(e+2)},
    \qquad e\in\mathbb N,
\]
for which
\[
    \sum_{e=0}^{\infty}\pi_e=1.
\]

\al{At round $n$, assign each active expert
$e\in A_n$ the weight}
\[
    \al{w_{e,n}
       :=\pi_e\exp(-\eta_n\widetilde L_{e,n}).}
\]
\al{The aggregate then predicts $1$ with probability}
\begin{equation}
    \al{q_n:=
    \frac{\sum_{e\in A_n}w_{e,n}P_e(X_0^{n-1})}
         {\sum_{e\in A_n}w_{e,n}}.}
    \label{eq:adaptive-pr-mixture}
\end{equation}

\al{Experts may enter within a dyadic epoch.
Their weights retain accumulated sleeping and
active losses; changing the learning rate
re-tempers rather than restarts the weights.}


At every round this mixture involves only finitely many active experts;
its total computability is verified in Theorem~\ref{thm:pr-aggregation}.
The complete potential analysis and the treatment of the learning-rate
changes are given in Appendix~\ref{app:exp-weights}.

\begin{theorem}[Regret of the primitive-recursive mixture predictor]
\label{thm:pr-aggregation}
Let $Q_{\mathrm{PR}}$ be the probabilistic predictor defined by
\eqref{eq:adaptive-pr-mixture}; that is,
\[
    Q_{\mathrm{PR}}(X_0^{n-1})=q_n .
\]
Then $Q_{\mathrm{PR}}$ is total computable. Moreover, for every
expert index $e$, every $X\in2^\omega$, and every $N>\tau_e$,
\begin{equation}
\begin{split}
    L_N(Q_{\mathrm{PR}},X)-L_N(P_e,X)
    \leq{}&
    \tau_e
    +\frac{\log(1/\pi_e)}{\eta_{N-1}}
    +\frac18\sum_{n<N}\eta_n .
\end{split}
\label{eq:adaptive-pr-regret}
\end{equation}
\al{Consequently, under the quadratic activation schedule
$\tau_e=e(e+1)/2$ and the dyadic learning rates,}
\begin{equation}
\begin{split}
    L_N(Q_{\mathrm{PR}},X)-L_N(P_e,X)
    \leq{}&
    \al{\frac{e(e+1)}{2}}\\
    &\al{+\sqrt N\left[
       \log((e+1)(e+2))
       +\frac{2+\sqrt2}{8}
    \right].}
\end{split}
\label{eq:pr-mixture-regret}
\end{equation}
In particular,
\[
    \limsup_{N\to\infty}
    \frac{L_N(Q_{\mathrm{PR}},X)}{N}
    \leq
    \limsup_{N\to\infty}
    \frac{L_N(P_e,X)}{N}
\]
for every $e$, and hence
\[
    \limsup_{N\to\infty}
    \frac{L_N(Q_{\mathrm{PR}},X)}{N}
    \leq e_{\mathrm{PR}}(X).
\]
Therefore $Q_{\mathrm{PR}}$ is a total computable PR-superpredictor.
\end{theorem}

\begin{proof}
The countable sleeping-expert potential inequality proved in
Appendix~\ref{app:exp-weights}, applied to the weights defining
\eqref{eq:adaptive-pr-mixture}, gives
\[
    L_N(Q_{\mathrm{PR}},X)-\widetilde L_{e,N}
    \leq
    \frac{\log(1/\pi_e)}{\eta_{N-1}}
    +\frac18\sum_{n<N}\eta_n .
\]
For $N>\tau_e$, the sleeping convention implies
\[
    \widetilde L_{e,N}
    =
    L_{\tau_e}(Q_{\mathrm{PR}},X)
    +L_N(P_e,X)-L_{\tau_e}(P_e,X).
\]
Since losses lie in $[0,1]$,
\[
    L_{\tau_e}(Q_{\mathrm{PR}},X)
    -L_{\tau_e}(P_e,X)
    \leq \tau_e,
\]
which proves \eqref{eq:adaptive-pr-regret}.

The dyadic estimates in Appendix~\ref{app:exp-weights} give
\[
    \eta_{N-1}^{-1}=O(\sqrt N),
    \qquad
    \sum_{n<N}\eta_n=O(\sqrt N),
\]
while
\[
   \al{\tau_e=e(e+1)/2},
    \qquad
    \log\frac1{\pi_e}
    =
    \log\bigl((e+1)(e+2)\bigr)
    =
    O\bigl(1+\log(e+2)\bigr).
\]
This proves \eqref{eq:pr-mixture-regret} and its asymptotic consequences.
Finally, on every finite input string only finitely many experts are active.
Their rational forecasts and cumulative losses are uniformly computable and
total.  The exponential weights are therefore computable reals.  Moreover,
for $n\in I_r$,
\[
    \sum_{\al{e\in A_n}}w_{e,n}
    \geq
    w_{0,n}
    =
    \pi_0\exp\bigl(-\eta_r\widetilde L_{0,n}\bigr)
    \geq
    \pi_0e^{-n}>0,
\]
so the normalization has an explicit computable positive lower bound.
Hence the quotient in \eqref{eq:adaptive-pr-mixture} can be computed to
arbitrary precision on every finite binary string.  Therefore
$Q_{\mathrm{PR}}$ is a total computable probabilistic predictor.
\end{proof}
\al{
\begin{remark}[Extension to real-valued PR forecasters]
The superprediction guarantee extends to every
$R\in\mathcal Q_{\mathrm{PR}}^{\mathrm{real}}$.
Let $\widehat R(\sigma,k)$ be its primitive-recursive
rational approximation from Definition~II.3, and define
\[
R^*(\sigma)=
\min\{1,\max\{0,\widehat R(\sigma,|\sigma|+1)\}\}.
\]
Then $R^*\in\mathcal Q_{\mathrm{PR}}^{\mathrm{rat}}$ and
\[
|L_N(R^*,X)-L_N(R,X)|
\le \sum_{n<N}2^{-n-1}<1.
\]
Consequently, Theorem~IV.3 applies to every real-valued
PR-computable forecaster, with an additional additive
constant of at most one in the regret bound.
\end{remark}
}
By Lemma~\ref{lem:superpredictor-not-pr},
$Q_{\mathrm{PR}}$ is necessarily not PR-computable.
The preceding construction is a total computable growing-expert implementation of
the countable-expert exponential-weights method.  The construction builds on the Aggregating Algorithm and expert-advice
frameworks, particularly their use of prior-weighted countable expert
classes, sleeping experts, and time-varying learning rates
\cite{CesaBianchiLugosi2006,vovk1990aggregating,vovk1998game,
ChernovKalnishkanZhdanovVovk2008}.  Here these ingredients are adapted to
the specific structure of the primitive-recursive class: experts are
introduced according to a computable activation schedule, only finitely many
experts are evaluated at each round, and their accumulated performance is
retained across dyadic epochs.  This yields a total computable mixture over
the full enumerated class together with the explicit regret bound above.

Two immediate consequences are worth recording.  Since every fixed
finite-state, finite-memory, and fixed context-tree predictor is primitive
recursive, $Q_{\mathrm{PR}}$ asymptotically competes with each such predictor
on every individual sequence.  Moreover, if $X$ itself is primitive
recursive, then the predictor
\[
    P_X(\sigma):=X_{|\sigma|}
\]
is primitive recursive and satisfies $L_N(P_X,X)=0$ for every $N$.
Therefore
\[
    \lim_{N\to\infty}\frac{L_N(Q_{\mathrm{PR}},X)}{N}=0.
\]
\al{The order $\sqrt N$ cannot be improved, for any choice of prior
weights. The class contains the constant predictors $0$ and $1$; against
fair coin flips every predictor has expected loss $N/2$, whereas the
better constant has expected loss $N/2-\Omega(\sqrt N)$. Hence for every
$N$ and every predictor there is a binary string of length $N$ on which
its regret against one of these two experts is
$\Omega(\sqrt N)$~\cite{CesaBianchiLugosi2006}.}

\begin{table*}[t]
\centering
\caption{\al{Online pseudocode for the effective constructions of
Theorem~\ref{thm:pr-aggregation} and Theorem~\ref{thm:beyond-pr}.}}
\label{tab:pseudocode}
\small
\begin{tabular}{@{}p{0.19\textwidth}p{0.76\textwidth}@{}}
\hline
\al{\textbf{Procedure}} & \al{\textbf{Steps}}\\
\hline

\al{PR-superpredictor} &
\al{
\emph{Initialization.} Fix an effective enumeration
$P_0,P_1,\ldots$ of $\mathcal Q_{\mathrm{PR}}^{\mathrm{rat}}$
(for Theorem~\ref{thm:beyond-pr}, use the ordering below),
the prior $\pi_e=1/((e+1)(e+2))$, and the activation
times $\tau_e=e(e+1)/2$.
Set $L\gets0$ and $A\gets\varnothing$.

\emph{Round $t=0,1,\ldots$,} with $x_0^{t-1}$ already observed:
\begin{enumerate}
\item[1.] \emph{Learning rate:}
If $t$ is the first round of epoch $I_r$, set
$\eta\gets2^{-r/2}$ and recompute
$w_e\gets\pi_e\exp(-\eta C_e)$ for every $e\in A$.

\item[2.] \emph{Activation:}
If $t=\tau_e$ for a new expert $e$, set
$C_e\gets L$,
$w_e\gets\pi_e\exp(-\eta L)$, and
$A\gets A\cup\{e\}$.

\item[3.] \emph{Prediction:}
Compute
\[
q_t=
\frac{\sum_{e\in A}w_eP_e(x_0^{t-1})}
{\sum_{e\in A}w_e}
\]
as a computable real and output it to the
requested precision.

\item[4.] \emph{Update,} after $x_t$ is revealed:
For every $e\in A$, compute
$d_e=\ell(P_e(x_0^{t-1}),x_t)$ and set
\[
C_e\gets C_e+d_e,\qquad
w_e\gets w_e\exp(-\eta d_e).
\]
Set $L\gets L+\ell(q_t,x_t)$.
The loss variables retain their exact mathematical values.
\end{enumerate}
}\\[1ex]

\al{Separating sequence} &
\al{
\emph{Initialization.} Fix an effective enumeration
$D_0,D_1,\ldots$ of the deterministic PR predictors.
Set $\sigma\gets\varnothing$, $k\gets0$, and let the
block buffer $W$ be empty.

\emph{Round $t=0,1,\ldots$:}
If $W$ is empty, set $k\gets k+1$ and enumerate
$u\in\{0,1\}^{k^3}$ lexicographically until,
for every $i<k$,
\[
\left|
\sum_{s<k^3}
\mathbf{1}\{D_i(\sigma u_0^{s-1})\neq u_s\}
-\frac{k^3}{2}
\right|\le k^2.
\]
Set $W\gets u$. Such a block exists by Hoeffding's
inequality and a union bound.
Output $X_t$, the first bit of $W$; remove it from
$W$ and append it to $\sigma$.
}\\[1ex]

\al{Expert ordering} &
\al{
\emph{Part of the initialization of the PR-superpredictor
in Theorem~\ref{thm:beyond-pr}.}
Let $T_k=\sum_{j=1}^{k}j^3$.
Given $k$, compute the finite prefix $X_0^{T_k-1}$
using the separating-sequence generator and encode
it as a finite lookup table.
Define $G_k(\sigma)=X_{|\sigma|}$ for $|\sigma|<T_k$,
using this table, and $G_k(\sigma)=0$ otherwise.
Each $G_k$ is PR, and its code is computable from $k$.
Given an effective enumeration $F_0,F_1,\ldots$ of
$\mathcal Q_{\mathrm{PR}}^{\mathrm{rat}}$, set
$P_{2j}=F_j$ and $P_{2j+1}=G_{j+1}$.
}\\
\hline
\end{tabular}
\end{table*}
\al{The construction is summarized in online form in
Table~\ref{tab:pseudocode}. The table also describes the
generator of the separating sequence and the computable ordering
of the experts used in Section~\ref{sec:separation}. That
ordering is a free choice within the construction: the
aggregation rule and its regret guarantee hold for every
effective enumeration.}

\section{PPM and Stationary Ergodic Sources}
\label{sec:ppm-ergodic}

\al{We establish Bayes optimality of the PR-superpredictor
on every Martin--L\"of random realization of a computable
stationary ergodic source. Building on Steifer's effective
Bayes lower bound~\cite{Steifer2019PhD} and the universal
PPM predictor of D\k{e}bowski and
Steifer~\cite{DebowskiSteifer2022}, we prove that PPM is
primitive recursive and extend the lower bound to
computable probabilistic predictors. Combined with
Theorem~\ref{thm:pr-aggregation}, these results yield
stationary-ergodic optimality.}
\al{The primitive-recursive membership argument applies to other
standard universal coding measures. For example, for each fixed
depth, the weighted probabilities of context-tree 
weighting~\cite{WillemsShtarkovTjalkens1995} are exact rationals
computed by a bounded recursion, so the corresponding thresholded
predictor belongs to $\mathcal P_{\mathrm{PR}}$.}

\subsection{PPM is primitive recursive}

\al{Prediction by Partial Matching (PPM) denotes a
family of context-based methods originating in
universal compression~\cite{ClearyWitten1984}.
We use the binary PPM measure of D\k{e}bowski
and Steifer~\cite{DebowskiSteifer2022}.
Unlike Ryabko's earlier construction, which used
Krichevsky--Trofimov ($+1/2$) smoothing, this
variant uses Laplace ($+1$) smoothing
\cite{DebowskiSteifer2022}.

Let $R_k$ denote its order-$k$ Markov component
and $R_{\mathrm{PPM}}$ the resulting mixture.
For a nonempty binary string
$\sigma=\sigma_0\cdots\sigma_{n-1}\in2^n$,
the order-$k$ component is}
\[
R_k(\sigma)
=
\begin{cases}
\displaystyle
2^{-k-1}
\prod_{i=k+1}^{n-1}
\frac{
   N(\sigma_{i-k}^{i}\mid\sigma_0^{i-1})+1
}{
   N(\sigma_{i-k}^{i-1}\mid\sigma_0^{i-2})+2
},
& k\leq n-2,\\[4mm]
2^{-n},
& k\geq n-1,
\end{cases}
\]
where \(N(u\mid v)\) denotes the number of occurrences of the finite
word \(u\) in \(v\), with overlaps allowed. 
\al{For the empty word, we use
$N(\varnothing\mid v)=|v|+1$,
as required by the binary Laplace normalization.}
 The total PPM measure is
\begin{align*}   
    R_{\mathrm{PPM}}(\sigma)
    &:=
    \sum_{k=0}^{\infty}
       \left(
          \frac{1}{k+1}-\frac{1}{k+2}
       \right)
       R_k(\sigma) \\ &=
    \sum_{k=0}^{\infty}
       \frac{R_k(\sigma)}{(k+1)(k+2)}.
\end{align*}
We set \(R_{\mathrm{PPM}}(\varnothing)=1\).  The induced deterministic
predictor is
\[
    P_{\mathrm{PPM}}(\sigma)
    :=
    \min\operatorname*{arg\,max}_{b\in\{0,1\}}
       R_{\mathrm{PPM}}(\sigma b),
\]
where the minimum fixes the tie-breaking convention.

\begin{lemma}
\label{lem:ppm-pr}
The predictor \(P_{\mathrm{PPM}}\) is primitive recursive.
\end{lemma}
\begin{proof}
\al{For $n=|\sigma|\geq1$, the identity
$R_k(\sigma)=2^{-n}$ for $k\geq n-1$
gives
\[
R_{\mathrm{PPM}}(\sigma)
=
\sum_{k=0}^{n-2}
\frac{R_k(\sigma)}{(k+1)(k+2)}
+\frac{2^{-n}}{n}.
\]
Every term is an exact rational expression
in bounded substring counts and products,
which are primitive recursive. Hence the
numerators and denominators of
$R_{\mathrm{PPM}}(\sigma0)$ and
$R_{\mathrm{PPM}}(\sigma1)$ are primitive
recursive. Their comparison by cross
multiplication proves the claim.}
\end{proof}

\subsection{Bayes optimality on stationary ergodic random points}

Let $\mu$ be a stationary binary measure, and let $e_\mu$ be the
infinite-past Bayes error defined through its natural two-sided extension in Section~\ref{subsec:stationary_bayes}. 
We first record the probabilistic extension of the effective source
prediction lower bound.

\begin{lemma}[Effective Bayes lower bound for probabilistic predictors]
\label{lem:probabilistic-bayes-lower-bound}
Let $\mu$ be a computable stationary ergodic binary measure, let
$X\in\operatorname{MLR}_\mu$, and let
\[
    Q:2^{<\omega}\to[0,1]
\]
be a total computable probabilistic predictor.  Then
\[
    \liminf_{N\to\infty}
       \frac{L_N(Q,X)}{N}
       \geq e_\mu.
\]
The conclusion relativizes: for any oracle $A$,
it holds for every total $A$-computable predictor
$Q$ and every $X\in\operatorname{MLR}_\mu^A$.
\end{lemma}

\begin{proof}[Proof sketch]
Write the realized loss as its conditional expectation plus a bounded
martingale difference.  A second-moment estimate on the square subsequence,
combined with effective Borel-Cantelli, makes the martingale term $o(N)$ on
every $\mu$-Martin-L\"of random point.  The conditional expected loss is at
least the finite-past Bayes error, whose Ces\`aro average converges to
$e_\mu$ by the effective L\'evy and Breiman theorems.  The treatment of
zero-measure cylinders, uniform effective openness, and relativization is
given in Appendix~\ref{app:effective-bayes-lower-bound}.
\end{proof}

\begin{theorem}[Primitive-recursive prediction is ergodically complete]
\label{thm:pr-ergodic-optimality}
Let \(\mu\) be a computable stationary ergodic binary measure and let
\(X\in\operatorname{MLR}_\mu\).  Then
\[
    e_{\mathrm{PR}}(X)=e_{\mathrm{comp}}(X)=e_\mu
\]
and the PR-superpredictor satisfies
\[
    \lim_{N\to\infty}
       \frac{L_N(Q_{\mathrm{PR}},X)}{N}
       =e_\mu.
\]
Equivalently, \(Q_{\mathrm{PR}}\) competes with \(P_{\mathrm{PPM}}\) on every
individual sequence and inherits its Bayes optimality on stationary
ergodic Martin-L\"of random points.
\end{theorem}

\begin{proof}

\al{
D\k{e}bowski and Steifer prove that a $1$-universal
measure $R$ induces a $1$-universal predictor provided
\[
-\log_2 R(b\mid\sigma)
=o\!\left(\sqrt{n/\ln n}\right),
\qquad |\sigma|=n,
\]
uniformly in $\sigma\in 2^n$ and $b\in\{0,1\}$
\cite{DebowskiSteifer2022}.
Our $R_{\mathrm{PPM}}$ is precisely their binary PPM measure
\cite{DebowskiSteifer2022}, with the same
induced predictor and tie-breaking rule.
They establish its $1$-universality and the bound
\[
-\log_2 R_{\mathrm{PPM}}(b\mid\sigma)
\leq 3\log_2(n+2)
=o\!\left(\sqrt{n/\ln n}\right)
\]
\cite{DebowskiSteifer2022}.
Thus their prediction theorem applies directly to
$P_{\mathrm{PPM}}$.
}

\al{Hence,
the PPM predictor satisfies}
\[
    \lim_{N\to\infty}
    \frac{L_N(P_{\mathrm{PPM}},X)}{N}
    =e_\mu.
\]
\al{Lemma~\ref{lem:ppm-pr} places this predictor
in $\mathcal P_{\mathrm{PR}}$, while the effective
deterministic Bayes lower bound
\cite{DebowskiSteifer2022}
applies to every total computable predictor.
Consequently,}
\[
    e_{\mathrm{PR}}(X)
    =e_{\mathrm{comp}}(X)=e_\mu.
\]
\al{The aggregation guarantee
of Theorem~\ref{thm:pr-aggregation} and
Lemma~\ref{lem:probabilistic-bayes-lower-bound}
now give}
\[
\al{
    e_\mu
    \leq
    \liminf_{N\to\infty}
       \frac{L_N(Q_{\mathrm{PR}},X)}{N}
    \leq
    \limsup_{N\to\infty}
       \frac{L_N(Q_{\mathrm{PR}},X)}{N}
    \leq e_\mu.
}
\]
\al{Hence the limit exists and equals $e_\mu$.}
\end{proof}

\al{
\begin{corollary}[Realized randomized predictions]
Under the assumptions of
Theorem~\ref{thm:pr-ergodic-optimality}, let
$U_n$ be independent $\mathrm{Unif}[0,1]$ variables,
independent of $X$, and set
$\widehat X_n=\mathbf 1\{U_n\leq
Q_{\mathrm{PR}}(X_0^{n-1})\}$.
Then, for every $X\in\operatorname{MLR}_\mu$,
\[
\frac1N\sum_{n<N}
\mathbf 1\{\widehat X_n\neq X_n\}
\longrightarrow e_\mu
\]
for product-Lebesgue-almost every realization of $(U_n)$.
\end{corollary}

\begin{proof}
For fixed $X$, the loss indicators are independent,
bounded, and have means
$\ell(Q_{\mathrm{PR}}(X_0^{n-1}),X_n)$.
Kolmogorov's strong law and
Theorem~\ref{thm:pr-ergodic-optimality}
give the result.
\end{proof}
}

\section{Bayes Optimality Beyond Stationarity}
\label{sec:pr-switching}

\al{We extend the stationary-ergodic optimality result
to nonstationary sequences obtained by
primitive-recursively interleaving finitely many
stationary ergodic sources. The switching schedule is unknown to the PR-superpredictor but can be
incorporated into a primitive-recursive comparator. The main difficulty is establishing optimality
against predictors that observe the entire
interleaved history. We then treat nonstationary hidden Markov sources and, as their one-state case, primitive-recursive signals observed in Bernoulli
noise.}

\subsection{Primitive-recursive interleavings of ergodic sources}
\label{subsec:pr-interleavings}

\al{We first specify the switching model and the
product-randomness assumptions under which the
statewise Bayes benchmark is valid.}

Let
\[
    \mathcal S=\{0,\ldots,|\mathcal{S}|-1\}
\]
be a finite state set and let
\[
    s:\mathbb N\longrightarrow\mathcal S
\]
be primitive recursive.  The sequence $s$ is regarded as an unknown
deterministic switching schedule: it is not supplied to the universal
aggregate, but it may be encoded in a primitive-recursive comparison
predictor.  For each $j$, put
\[
    N_j(n):=\bigl|\{t<n:s(t)=j\}\bigr|.
\]

For each state $j$, let $\mu_j$ be a computable stationary ergodic
binary measure, and let $Y^{(j)}\in2^\omega$ be its state-specific source
stream.  We assume that the streams are jointly Martin-L\"of random for the
computable product measure
\[
    \boldsymbol{\mu}:=\bigotimes_{j}\mu_j.
\]
The observed sequence $X$ is defined by
\begin{equation}
    X_n
       :=
    Y^{(s(n))}_{N_{s(n)}(n)}.
    \label{eq:pr-interleaving}
\end{equation}
Thus the stream associated with state $j$ advances only when $s(n)=j$.
A related hierarchy appears in Markovian bandit problems.  In the classical rested model of Gittins and Jones~\cite{GittinsJones1974}, an unselected arm remains frozen, whereas in Whittle's restless model~\cite{Whittle1988}, passive arms continue to evolve.  This distinction
is separate from state observability.  In hidden Markov bandits, the current
state is not observed directly, and decisions must instead be based on a
filtered belief state~\cite{KrishnamurthyWahlberg2009,
YeminiLeshemSomekhBaruch2021}.  The Markov prediction problem considered
above is analogous to the fully observed case, since the most recent symbol is contained in the observed prefix. Prediction for a hidden Markov source involves the additional problem of latent-state inference. Section~\ref{subsec:nonstationary-hmm} treats this case: the aggregate attains the best performance achievable from the observations alone.

The induced law of $X$ is generally nonstationary, and the state
frequencies $N_j(N)/N$ need not converge.

For each $j$, write
\[
    e_j:=e_{\mu_j}
\]
for the infinite-past Bayes error of $\mu_j$.

\begin{theorem}[Primitive-recursive switching among ergodic sources]
\label{thm:pr-switched-ergodic}
Under the assumptions above,
\begin{equation}
    L_N(Q_{\mathrm{PR}},X)
       =
    \sum_{j}N_j(N)e_j+o(N).
    \label{eq:pr-switched-optimality}
\end{equation}
More generally, every computable probabilistic predictor
$Q:2^{<\omega}\to[0,1]$ satisfies
\begin{equation}
    \liminf_{N\to\infty}
    \frac{
       L_N(Q,X)
       -
       \sum_{j}N_j(N)e_j
    }{N}
    \geq0.
    \label{eq:pr-switched-lower}
\end{equation}
Consequently, the PR-superpredictor is asymptotically
optimal for the nonstationary individual sequence $X$.

If the state frequencies converge,
\[
    \frac{N_j(N)}{N}\longrightarrow\rho_j,
\]
then
\[
    \frac{L_N(Q_{\mathrm{PR}},X)}{N}
       \longrightarrow
    \sum_{j}\rho_j e_j.
\]
\end{theorem}

\begin{proof}[Proof sketch]
A state-aware primitive-recursive comparator runs PPM separately on each
subsequence selected by the schedule, so its loss decomposes as
$\sum_{j}N_j(N)e_j+o(N)$.  The superprediction guarantee gives the matching
upper bound for $Q_{\mathrm{PR}}$.  For the lower bound, product randomness
and the generalized van Lambalgen theorem make each component random relative
to the join of the others; the relativized Bayes lower bound can then be
applied componentwise and summed.  The complete argument appears in
Appendix~\ref{app:pr-switching-proof}.
\end{proof}

\begin{remark}
The product assumption is essential for the benchmark in
\eqref{eq:pr-switched-optimality}.  If the state-specific streams are
dependent, observations collected in one state may reveal future symbols in
another state, and the weighted sum of the marginal Bayes errors need not be
optimal.  The same theorem remains valid under the weaker requirement that,
at each activation of state $j$, the conditional law of the next
state-$j$ symbol given the complete interleaved past agrees with its
conditional law given the past of the $j$-th stream alone.
\end{remark}

\subsection{Nonstationary hidden Markov sources}
\label{subsec:nonstationary-hmm}

\al{We next consider finite-state hidden Markov sources with
primitive-recursively time-varying parameters. Their states are
unobserved, and neither the transition matrices nor the emission
probabilities need converge, be periodic, or satisfy mixing
assumptions. This includes nonstationary variants of the
Gilbert--Elliott channel.}

\al{Fix $J\geq1$ and let $\lambda$ be a rational initial distribution
on the state set $\{1,\ldots,J\}$. For each time $n\geq0$ let
$T_n=(T_n(j,k))_{j,k=1}^J$ be a $J\times J$ stochastic matrix with
rational entries, and let
$\theta_n=(\theta_n(j))_{j=1}^J\in(\mathbb Q\cap[0,1])^J$.
Throughout, the subscript $n$ is the time index and the arguments
$j,k$ are states: $T_n(j,k)$ is the probability of moving from state
$j$ at time $n$ to state $k$ at time $n+1$, and $\theta_n(j)$ is the
probability of emitting the symbol $1$ in state $j$ at time $n$. We
assume that $(n,j,k)\mapsto T_n(j,k)$ and $(n,j)\mapsto\theta_n(j)$
are primitive recursive, as exact rational representations. The hidden
states $\xi_n$ and observations $X_n\in\{0,1\}$ satisfy
$\xi_0\sim\lambda$ and, for every $n\geq0$ and $1\leq k\leq J$,}
\[
\al{
\begin{aligned}
\mathbb P\bigl(X_n=1\mid\xi_0^n,X_0^{n-1}\bigr)
&=\theta_n(\xi_n),\\
\mathbb P\bigl(\xi_{n+1}=k\mid\xi_0^n,X_0^n\bigr)
&=T_n(\xi_n,k).
\end{aligned}
}
\]
\al{Thus the emission and transition laws at step $n$ depend on the
time $n$ and on the current hidden state $\xi_n$, and on nothing
else: given $\xi_n$, they do not depend on earlier states or on the
observations. The parameter sequences $(T_n)_{n\geq0}$ and
$(\theta_n)_{n\geq0}$ are fixed in advance and are unknown to the
PR-superpredictor.}

\al{Let $\mu$ be the induced law of the observed sequence. For
$\mu([\sigma])>0$, define the conditional Bayes error}
\[
\al{
\beta_\mu(\sigma)
:=\frac{\min\{\mu([\sigma0]),\mu([\sigma1])\}}
{\mu([\sigma])},
}
\]
\al{and put $\beta_\mu(\sigma)=0$ on null cylinders. The
cumulative conditional Bayes benchmark is}
\[
\al{
B_N(\mu,X):=\sum_{n<N}\beta_\mu(X_0^{n-1}).
}
\]
\al{Its normalized value need not converge. Note that
$\beta_\mu(\sigma)$ conditions on the observed prefix only. It is the
error of the best prediction of $X_n$ by an observer who knows $\lambda$
and the entire parameter sequence $(T_n,\theta_n)_{n\geq0}$ but does not
see the states. We refer to $B_N(\mu,X)$ as the \emph{observation-only}
(filtered) Bayes benchmark.}

\begin{theorem}[Universal prediction for nonstationary hidden Markov sources]
\label{thm:pr-nonstationary-hmm}
\al{Let $\mu$ be an observation law as above and let
$X\in\operatorname{MLR}_\mu$.}
\begin{enumerate}
\item[(i)] \al{Every total computable probabilistic predictor $Q$
satisfies}
\[
\al{
L_N(Q,X)\geq B_N(\mu,X)-\sqrt{N\log N}
}
\]
\al{for all sufficiently large $N$.}
\item[(ii)] \al{The PR-superpredictor satisfies}
\[
\al{
\bigl|L_N(Q_{\mathrm{PR}},X)-B_N(\mu,X)\bigr|
=O\!\left(\sqrt{N\log N}\right).
}
\]
\end{enumerate}
\al{Thus $Q_{\mathrm{PR}}$, which knows neither the parameters nor the
states, matches the observation-only Bayes benchmark to within
$O(\sqrt{N\log N})$, and no total computable predictor improves on that
benchmark by more than $\sqrt{N\log N}$.}
\end{theorem}

\begin{proof}
\al{Let $\mathbb P$ be the joint law of states and observations,
with marginal $\mu$. For $n=|\sigma|$ put
$\alpha_\sigma(j):=\mathbb P(X_0^{n-1}=\sigma,\xi_n=j)$,
so that $\alpha_\varnothing=\lambda$. Writing
$\theta_{n,1}=\theta_n$ and $\theta_{n,0}=1-\theta_n$,
the forward recursion gives}
\[
\al{
\begin{aligned}
\mu([\sigma b])
&=\sum_{j=1}^J\alpha_\sigma(j)\theta_{n,b}(j),\\
\alpha_{\sigma b}(k)
&=\sum_{j=1}^J\alpha_\sigma(j)\theta_{n,b}(j)T_n(j,k).
\end{aligned}
}
\]
\al{Thus $\sigma\mapsto\mu([\sigma])$ is primitive-recursive
and rational-valued. In particular, the exact Bayes rule}
\[
\al{
P_\mu(\sigma):=\mathbf 1\{\mu([\sigma1])>\mu([\sigma0])\}
}
\]
\al{is primitive recursive; its conditional expected loss on each
positive-measure prefix equals $\beta_\mu(\sigma)$.}

\al{For any total computable probabilistic predictor $Q$, let
$c_n$ be its conditional expected loss given $X_0^{n-1}$,
assigned a fixed value on null cylinders, and put
$D_n:=\ell(Q(X_0^{n-1}),X_n)-c_n$. Exact rational cylinder
probabilities make null cylinders decidable and $c_n$ uniformly
computable. The $D_n$ are martingale differences whose conditional
ranges have length at most $1$. Hence, by the Hoeffding--Azuma
inequality for martingale differences confined to predictable intervals
\cite[Sec.~A.1.3]{CesaBianchiLugosi2006}, for $N\geq2$,}
\[
\al{
\mu\left(\left|\sum_{n<N}D_n\right|>\sqrt{N\log N}\right)
\leq 2N^{-2}.
}
\]
\al{The cruder bound based on $|D_n|\leq1$ alone is not summable at
this threshold. The events are uniformly effectively open; effective
Borel--Cantelli \cite{DebowskiSteifer2022} therefore gives, on every
$\mu$-Martin--L\"of random point,
$\bigl|\sum_{n<N}D_n\bigr|\leq\sqrt{N\log N}$ for all sufficiently
large $N$. Since $c_n\geq\beta_\mu(X_0^{n-1})$, this proves~(i).
For $P_\mu$, where $c_n=\beta_\mu(X_0^{n-1})$, the same estimate gives}
\[
\al{
L_N(P_\mu,X)=B_N(\mu,X)+O(\sqrt{N\log N}).
}
\]
\al{Finally, $P_\mu\in\mathcal P_{\mathrm{PR}}$, so
Theorem~\ref{thm:pr-aggregation} gives
$L_N(Q_{\mathrm{PR}},X)\leq L_N(P_\mu,X)+O(\sqrt N)$.
Part~(i), applied to the total computable predictor
$Q_{\mathrm{PR}}$, gives the matching lower bound and hence~(ii).}
\end{proof}

\begin{remark}[The benchmark is observation-only]
\label{rem:hmm-benchmark}
\al{Theorem~\ref{thm:pr-nonstationary-hmm} compares $Q_{\mathrm{PR}}$
with predictors that see the observations and nothing else. Every
predictor in this paper, primitive recursive or not, is a function of the
observed prefix.
An observation-only predictor may maintain an estimate
of the hidden state, including its posterior distribution
when the parameters are known, but it does not observe
the realized state $\xi_n$.  The Bayes rule $P_\mu$ is exactly of this kind: it carries a
state estimate, not the state.}

\al{The benchmark $B_N$ should not be confused with the loss of a
\emph{state-informed} predictor, which is told $\xi_n$ before
predicting $X_n$. The optimal such predictor outputs the more likely
symbol in the revealed state and has conditional expected loss}
\[
\al{
\kappa_n:=\min\{\theta_n(\xi_n),\,1-\theta_n(\xi_n)\}.
}
\]
\al{Since $\mathbb P(X_n=1\mid X_0^{n-1})
=\E[\theta_n(\xi_n)\mid X_0^{n-1}]$ and $t\mapsto\min\{t,1-t\}$
is concave, Jensen's inequality gives}
\[
\al{
\E\bigl[\kappa_n\mid X_0^{n-1}\bigr]
\leq\beta_\mu(X_0^{n-1})
\qquad\text{almost surely,}
}
\]
\al{and hence
$\E\sum_{n<N}\kappa_n\leq\E\,B_N(\mu,X)$.}

\al{The gap can be linear in $N$. Take $J=2$, the uniform initial
distribution, $T_n(j,k)=1/2$ for all $n,j,k$, and
$\theta_n=(\delta,1-\delta)$ with rational $0<\delta<1/2$. The
observed process is then a sequence of independent fair bits, so
$\beta_\mu\equiv1/2$ and $B_N(\mu,X)=N/2$, whereas the state-informed
predictor has expected loss $\delta N$. Here the observed prefix carries
no information about the current state, and every total computable
predictor, including $P_\mu$ and $Q_{\mathrm{PR}}$, has loss
$N/2+O(\sqrt{N\log N})$ on every Martin-L\"of random realization.}

\al{The difference between the two benchmarks is the price of the
hidden state. It is a property of the source and is paid by every
observation-only predictor.
Theorem~\ref{thm:pr-nonstationary-hmm} makes no claim relative to the
state-informed benchmark. Its content is that ignorance of the parameters, in addition to
ignorance of the state, costs only $O(\sqrt{N\log N})$.}
\end{remark}

\begin{remark}[Nonstationary Gilbert--Elliott channels]
\al{
The binary error process of a Gilbert--Elliott channel is a
two-state hidden Markov source \cite{Gilbert1960,Elliott1963},
so Theorem~\ref{thm:pr-nonstationary-hmm} covers channels with
primitive-recursively time-varying transition and error
probabilities. The PR-superpredictor predicts the error sequence
with asymptotically Bayes-optimal cumulative loss, in the
observation-only sense of Remark~\ref{rem:hmm-benchmark},
without knowing the channel state or parameters.
This is a prediction result, not a channel-coding or
capacity theorem.

If the channel state is revealed before each error bit
and the error probabilities are time invariant, the
problem is much simpler: the Bayes rule is one of the
four maps from states to $\{0,1\}$, and aggregating
these four rules suffices. With time-varying error
probabilities, however, every fixed rule can incur
linear excess loss, so a richer comparator class
may be needed.
}
\end{remark}

\al{
\subsection{Bayes optimality for noisy PR signals}
\label{subsec:noisy-pr-bayes}

\begin{proposition}
\label{prop:noisy-pr-bayes}
\al{Let $c\in2^\omega$ be primitive recursive, let
$0<\varepsilon<1/2$ be rational, and let $\nu_{c,\varepsilon}$ be
the law of $X_n=c_n\mathbin\oplus Z_n$, where the $Z_n$ are
independent Bernoulli-$\varepsilon$ random variables. Then every
$X\in\operatorname{MLR}_{\nu_{c,\varepsilon}}$ satisfies
\[
L_N(Q_{\mathrm{PR}},X)=\varepsilon N+O(\sqrt{N\log N}),
\qquad
e_{\mathrm{PR}}(X)=e_{\mathrm{comp}}(X)=\varepsilon,
\]
and $L_N(Q,X)\geq\varepsilon N-\sqrt{N\log N}$ for all
sufficiently large $N$, for every total computable probabilistic
predictor $Q$.}
\end{proposition}

\begin{proof}
\al{This is the case $J=1$ of
Theorem~\ref{thm:pr-nonstationary-hmm}, with
$\theta_n=1-\varepsilon$ if $c_n=1$ and $\theta_n=\varepsilon$
otherwise, so that $B_N(\nu_{c,\varepsilon},X)=\varepsilon N$.
The Bayes rule is the primitive-recursive position predictor
$P_c(\sigma)=c_{|\sigma|}$, whose loss is
$\varepsilon N+O(\sqrt{N\log N})$ by the proof of the theorem.
Together with the lower bound, this gives
$e_{\mathrm{PR}}(X)=e_{\mathrm{comp}}(X)=\varepsilon$.}
\end{proof}
}

\al{
\section{Computational Separations}
\label{sec:separation}

\al{We now turn from Bayes optimality to computational separations: sequences on which every predictor in a
given class has asymptotic loss $1/2$, while a PR-superpredictor has vanishing loss.}

\al{Our main separation result, Theorem~\ref{thm:beyond-pr}, constructs a
computable sequence on which every primitive-recursive forecaster has
asymptotic loss $1/2$, while a PR-superpredictor, obtained from the same
aggregation rule with a suitable computable ordering of its experts, has
loss $O(\sqrt N\log N)$. Before it, we show a simpler level-by-level version
that follows from Di~Paola's subrecursive random sets, with finite-state
prediction as the base case. It is presented to provide insight into the main theorem.
}

\al{
\subsection{A strict hierarchy of primitive-recursive prediction}
\label{subsec:pr-hierarchy-separation}

Di~Paola's theorem on random sets in subrecursive hierarchies
\cite[Thm.~4.2]{DiPaola1969} establishes a strict hierarchy for
Church's notion of selection-rule randomness. Let $\mathcal E^n$
denote the Grzegorczyk classes, with $\mathcal E^3$ the elementary
functions and $\bigcup_n\mathcal E^n=\mathrm{PR}$
\cite{Grzegorczyk1953}. For each $n\geq3$, Di~Paola constructed a
binary sequence $X^{(n)}$ whose characteristic function belongs to
$\mathcal E^{n+1}$, yet which is Church-random relative to
$\mathcal E^n$. Thus the sequence is computationally accessible at
the next level but no selection rule at level $n$ can extract an
infinite subsequence with limiting one-frequency different from
$1/2$. We translate this result into prediction loss and then apply
the PR-superprediction theorem.

For completeness, code the prefix $\sigma=a_0\cdots a_{t-1}$ by
$b_\varnothing=1$ and $b_{\sigma a}=2b_\sigma+a$, as in
\cite{DiPaola1969}. A sequence $X$ is
\emph{$\mathcal E^n$-random in the sense of Church} when its
one-frequency tends to $1/2$ and the same holds along every infinite
subsequence of positions selected by a $\{0,1\}$-valued unary rule
$g\in\mathcal E^n$ applied to $b_{X_0^{t-1}}$ before observing $X_t$.
Let $\mathcal Q_{\mathcal E^n}^{\mathrm{rat}}$ consist of the
forecasters $R_{a,b}$ from Section~\ref{subsec:pr-benchmark} for
which the two integer-valued functions $a,b$, acting on the prefix
code, belong to $\mathcal E^n$. Set
\[
  e_{\mathcal E^n}(X)
    :=\inf_{Q\in\mathcal Q_{\mathcal E^n}^{\mathrm{rat}}}
      \limsup_{N\to\infty}\frac{L_N(Q,X)}{N}.
\]

\begin{lemma}[Church randomness and conditional expected prediction loss]
\label{lem:church-defeats-forecasters}
For $n\geq3$, if $X$ is $\mathcal E^n$-random in the sense of
Church, then
\[
  \lim_{N\to\infty}\frac{L_N(Q,X)}{N}=\frac12
  \qquad\text{for every }Q\in\mathcal Q_{\mathcal E^n}^{\mathrm{rat}}.
\]
\end{lemma}

\begin{proof}
Fix $Q=R_{a,b}$ in the stated class and an integer $K\geq2$.
Partition $[0,1]$ into intervals
$I_j=[j/K,(j+1)/K)$ for $j<K-1$ and
$I_{K-1}=[(K-1)/K,1]$. Let
$S_j=\{t:Q(X_0^{t-1})\in I_j\}$.
Membership in each $S_j$ is decided from the prefix code by
integer comparisons involving $a$ and $b$. At points with
$a+b=0$, the prescribed value $Q=1/2$ is handled by the same
finite case distinction. Since $\mathcal E^n$ contains the
elementary arithmetic operations and is closed under composition
and bounded minimization for $n\geq3$, each such indicator is an
$\mathcal E^n$ selection rule.

Write $M_j(N)=|S_j\cap[0,N)|$ and let $H_j(N)$ count the ones of
$X$ at those positions. Church randomness gives
\[
   H_j(N)=\tfrac12 M_j(N)+o(N)
\]
for every $j$: if $S_j$ is infinite the error is
$o(M_j(N))$, and if $S_j$ is finite it is $O(1)$.
On $S_j$, replace $Q$ by the constant $q_j=j/K$.
Because $\ell(q,x)$ is $1$-Lipschitz in $q$, this changes the
loss accumulated on $S_j$ by at most $M_j(N)/K$, hence the total
loss by at most $N/K$. The constant prediction incurs
\[
  q_j\bigl(M_j(N)-H_j(N)\bigr)
  +(1-q_j)H_j(N)=\tfrac12 M_j(N)+o(N).
\]
Summation over the fixed number $K$ of bins therefore yields
\[
   \limsup_{N\to\infty}
   \left|\frac{L_N(Q,X)}{N}-\frac12\right|
      \leq\frac1K.
\]
Since $K$ is arbitrary, the asserted limit follows.
\end{proof}

\begin{theorem}[Uniform separation at every Grzegorczyk level]
\label{thm:grzegorczyk-prediction-hierarchy}
For every $n\geq3$, there exists a primitive-recursive binary
sequence $X^{(n)}$, with characteristic function in
$\mathcal E^{n+1}$, such that
\[ \lim_{N\to\infty}\frac{L_N(Q,X^{(n)})}{N}=\frac12
  \qquad\text{simultaneously for every }
  Q\in\mathcal Q_{\mathcal E^n}^{\mathrm{rat}},
\]
while
\[
\lim_{N\to\infty}\frac{L_N(Q_{\mathrm{PR}},X^{(n)})}{N}=0.
\]
Consequently,
\[
  e_{\mathcal E^n}(X^{(n)})=\frac12,
  \qquad
  e_{\mathcal E^{n+1}}(X^{(n)})
   =e_{\mathrm{PR}}(X^{(n)})=0.
\]
In particular, $X^{(n)}\notin\mathcal E^n$.
\end{theorem}

\begin{proof}
Let $X^{(n)}$ be the characteristic sequence supplied by
Di~Paola's theorem \cite{DiPaola1969}.
The preceding lemma gives the first asserted limit for every
level-$n$ forecaster on the \emph{same} sequence.
The position predictor
$P_{X^{(n)}}(\sigma)=X^{(n)}_{|\sigma|}$ is in
$\mathcal E^{n+1}$: the prefix-length operation is elementary,
and $\mathcal E^{n+1}$ is closed under composition.
It makes no errors on $X^{(n)}$. In particular it is a
primitive-recursive expert, so Theorem~\ref{thm:pr-aggregation}
ensures $L_N(Q_{\mathrm{PR}},X^{(n)})=O(\sqrt N)$.
If $X^{(n)}$ belonged to $\mathcal E^n$, the same position
predictor would belong to $\mathcal E^n$ and contradict the
first limit. The benchmark equalities follow.
\end{proof}

\begin{remark}[Uniformity and attribution]
\label{rem:hierarchy-new}
Di~Paola established the existence of the level-$n$ Church-random
witnesses. The lemma transfers their randomness to conditional
expected zero--one loss for rational level-$n$ forecasters.
The aggregation consequence is that the \emph{same} computable
PR-superpredictor, without knowing $n$, achieves vanishing
asymptotic loss on the witnesses for every fixed level $n$.
The witness $X^{(n)}$ depends on $n$; no single primitive-recursive
sequence can defeat every PR level, because its perfect position
predictor belongs to some finite level. Thus the uniformity is
within each level and in the choice of the master, not in the
choice of a common witness or a common finite-time regret bound.
\end{remark}
}}

\al{
A finite-state predictor is a finite automaton, with finite state set
$\mathcal K$, whose output,
before reading the next bit, is either deterministic
($g:\mathcal K\to\{0,1\}$) or probabilistic
($q:\mathcal K\to[0,1]$).
Let $\mathcal P_{\mathrm{FS}}$ denote the deterministic class.

\begin{corollary}[Finite-state separation]
\label{cor:fs-base-case}
The witness $X=X^{(3)}$ of
Theorem~\ref{thm:grzegorczyk-prediction-hierarchy}
is primitive recursive and normal. Every finite-state
predictor, deterministic or randomized, has asymptotic
loss $1/2$ on $X$, whereas
$L_N(Q_{\mathrm{PR}},X)/N\to0$. Consequently,
\[
e_{\mathrm{FS}}(X)=1/2,\qquad e_{\mathrm{PR}}(X)=0.
\]
\end{corollary}

\begin{proof}
The rules selecting visits to a finite-automaton state,
or positions following a prescribed finite word,
are elementary. Thus $\mathcal E^3$-Church randomness
implies unbiased frequencies on every infinite such
subsequence. The former gives loss $1/2$ for each
finite-state predictor, including randomized outputs;
the latter implies normality by induction on word
length. The remaining claims follow from the theorem.
\end{proof}

We now prove the main separation result: the entire primitive-recursive class is separated from a PR-superpredictor on a single computable
sequence, with the maximal gap. The separating sequence is then computable but
no longer primitive recursive.

\begin{theorem}[Separation from the full PR class]
\label{thm:beyond-pr}
There exist a computable binary sequence $X$ and an effective enumeration
of $\mathcal Q_{\mathrm{PR}}^{\mathrm{rat}}$ whose PR-superpredictor
$Q_{\mathrm{PR}}^{\mathrm{ord}}$, defined by \eqref{eq:adaptive-pr-mixture},
satisfies
\[
\lim_{N\to\infty}\frac{L_N(R,X)}{N}=\frac12
\quad\text{for every }R\in\mathcal Q_{\mathrm{PR}}^{\mathrm{rat}},
\qquad
L_N(Q_{\mathrm{PR}}^{\mathrm{ord}},X)=O(\sqrt N\log(N+2)).
\]
Consequently, $e_{\mathrm{PR}}(X)=1/2$ and
$e_{\mathrm{comp}}(X)=0$.
\end{theorem}

\begin{proof}
\emph{Construction of $X$.}
Enumerate the deterministic PR predictors as $D_0,D_1,\ldots$.
Construct $X=W_1W_2\cdots$, with $|W_k|=k^3$.  Given the
previous prefix $\sigma$, choose the lexicographically first block
$u\in\{0,1\}^{k^3}$ for which every $i<k$ satisfies
\[
\left|\sum_{t<k^3}\mathbf 1\{D_i(\sigma u_0^{t-1})\ne u_t\}
-\frac{k^3}{2}\right|\le h_k,
\qquad h_k=k^2.
\]
For each $D_i$ and fixed $\sigma$, the mapping from $u$ to its
prediction-error string is a bijection. Thus, for uniform $u$, the
error count has distribution $\mathrm{Binomial}(k^3,1/2)$.
Hoeffding's inequality and a union bound show that the fraction
of blocks violating at least one constraint is at most
$2k\exp(-2h_k^2/k^3)\le2k\exp(-2k)<1$.
The bounded search therefore terminates. Its successive stages
are uniformly computable, so $X$ is computable. For each fixed
$D_i$, all sufficiently late blocks have error
$k^3/2+O(k^2)$; since
$\sum_{k\le K}k^2=o(\sum_{k\le K}k^3)$ and
$K^3=o(\sum_{k<K}k^3)$, its loss fraction tends to $1/2$
at every horizon, including within blocks.

\emph{Rational forecasters.}
This extends to rational PR forecasters without an additional
randomness argument. For fixed $R\in\mathcal Q_{\mathrm{PR}}^{\mathrm{rat}}$
and $m\ge1$, define
\[
R^{(m)}(\sigma)=2^{-m}\lfloor2^mR(\sigma)\rfloor
=2^{-m}\sum_{j=1}^{2^m}
\mathbf1\{R(\sigma)\ge j2^{-m}\}.
\]
Each threshold rule in this sum is a deterministic PR predictor.
Since the loss is affine in the forecast, $L_N(R^{(m)},X)$ is the
average of the losses of these $2^m$ predictors, and therefore
$L_N(R^{(m)},X)/N\to1/2$. Since
$|R^{(m)}(\sigma)-R(\sigma)|\le2^{-m}$ and the loss is
$1$-Lipschitz in the forecast, letting $m\to\infty$ proves
$L_N(R,X)/N\to1/2$.

\emph{The master.}
Set $T_k=\sum_{j\le k}j^3$ and $T_0=0$. For each $k$, let $G_k$ be the
position predictor that outputs $X_{|\sigma|}$ for $|\sigma|<T_k$ and $0$
thereafter. It is a finite table, hence primitive recursive, and its code
is computable from $k$; moreover $L_N(G_k,X)=0$ for $N\le T_k$. Let
$F_0,F_1,\ldots$ enumerate $\mathcal Q_{\mathrm{PR}}^{\mathrm{rat}}$ and put
$P_{2j}=F_j$, $P_{2j+1}=G_{j+1}$. This is an effective complete
enumeration, and Theorem~\ref{thm:pr-aggregation} applies to its mixture
$Q_{\mathrm{PR}}^{\mathrm{ord}}$. The expert $G_k$ has index $2k-1$ and
activation time $2k^2-k$. For all sufficiently large $k$ and $T_{k-1}<N\le T_k$,
 \eqref{eq:pr-mixture-regret} gives
\[
L_N(Q_{\mathrm{PR}}^{\mathrm{ord}},X)
\le 2k^2-k+\sqrt N\left[
\log(2k(2k+1))+\frac{2+\sqrt2}{8}\right]=O(\sqrt N\log(N+2)),
\]
since $k=\Theta(N^{1/4})$. Finally, the position predictor of $X$ is
computable with zero loss, which gives $e_{\mathrm{comp}}(X)=0$.
\end{proof}

\begin{remark}[Role of the expert ordering]
\label{rem:beyond-pr-scope}
The sequence $X$ does not depend on the ordering of the experts, but the
bound for the master does: the stage predictors $G_k$ must be activated
before the blocks they predict. Therefore, the theorem concerns a
suitable computable ordering, not an arbitrary fixed enumeration. Each
$G_k$ is primitive recursive, whereas the uniform map
$(k,\sigma)\mapsto G_k(\sigma)$ is not PR: otherwise
$n\mapsto G_{n+1}(0^n)=X_n$ would be PR, contrary to the
separation just proved.
\end{remark}

Geometrically, for each deterministic predictor the bad blocks are
preimages, under its causal error-string bijection, of two Hamming
balls of radius approximately $k^3/2-h_k$. Each preimage has the
volume of the corresponding Hamming ball. The union bound shows that
the $k$ pairs of bad sets cannot cover the binary cube.
Computable sets that are random in the sense of Church
with respect to all primitive-recursive selection rules were constructed
earlier, by different means, by Levin, Minsky, and Silver; see
Di~Paola~\cite{DiPaola1969}.

\begin{corollary}[Strict inclusions]
\label{cor:strict-inclusions}
\[
\mathcal P_{\mathrm{FS}}
\subsetneq\mathcal P_{\mathrm{PR}}
\subsetneq\mathcal P_{\mathrm{comp}}.
\]
\end{corollary}

\begin{proof}
Finite-state predictors are primitive recursive by iteration of
their finite transition tables, and the first inclusion is strict
by Corollary~\ref{cor:fs-base-case}. The second is strict by
Theorem~\ref{thm:beyond-pr}, since its computable sequence has a
perfect computable position predictor but no PR predictor with
asymptotic loss below $1/2$.
\end{proof}
}
\section{Scope and Further Directions}
\label{sec:discussion}
\al{
The primitive-recursive class is an effective reference class rather
than a guarantee of computational efficiency. Only finitely many experts
are evaluated at each round, and every activated computation halts, but
their resource requirements may be enormous. Our results establish
computability and asymptotic universality, not practical running-time
bounds. Restricting the benchmark to explicitly time- or space-bounded
experts would allow computational complexity to be incorporated directly
into the prediction guarantees.

An important distinction concerns nested model classes whose complexity
increases with the prediction horizon. Our regret guarantees hold against
each fixed PR forecaster, not necessarily against the best model selected
anew at each horizon. Growing context-tree classes and sublinear state
budgets, as studied by Ziv and Merhav~\cite{ZivMerhav2007}, require control
of the additional model-complexity and adaptation penalties. Such
extensions would connect computational universality with statistical model
selection.

Related questions arise for sources exhibiting Hilberg-type scaling or the
long-range information structure of Santa Fe
processes~\cite{Debowski2021book}. Their information-growth properties are
not characterized by our zero--one-loss guarantees. Similarly, trained
language models with fixed weights and finite-precision,
primitive-recursively bounded inference are PR forecasters, so the
PR-superpredictor competes with each of them asymptotically. Such models
form nested families of increasing complexity, so aggregating them
effectively requires model complexity as an additional parameter,
favoring simpler models over more complex ones, as in the nested-class
setting discussed above. Our results do not characterize their
computational efficiency or scaling behavior.

Other extensions include nonstationary HMMs beyond exact rational PR
transition and emission laws, noncomputable ergodic source measures, and
more general loss functions. Logarithmic loss is especially natural, as
noted in the Conclusions.

Finally, optimal prediction error does not necessarily imply that the
resulting innovation sequence possesses strong randomness properties. The
relationship between prediction and innovation randomness is studied in
the companion paper~\cite{leshem2026sequential}.
}
\section{Conclusions}
\al{
We developed a computable PR-superpredictor using
exponential weighting and delayed activation over an
effective enumeration of primitive-recursive forecasters.
It achieves explicit sublinear regret against every
rational-valued PR forecaster on every individual binary
sequence. The method extends to effectively enumerable
classes of uniformly computable total predictors, although
no PR predictor can achieve the same universal guarantee.

The individual-sequence guarantee yields Bayes optimality
on every Martin-L\"of random realization of a computable
stationary ergodic binary source, using the primitive
recursiveness of PPM. This optimality extends to PR
interleavings of independent ergodic sources and to
nonstationary finite-state hidden Markov models with
primitive-recursively time-varying rational parameters,
including the error processes of nonstationary
Gilbert--Elliott channels.

We also established two complementary computational
separations. A single computable sequence forces every
PR forecaster to incur asymptotic loss $1/2$, while a
PR-superpredictor with a suitable effective expert
ordering achieves $O(\sqrt N\log N)$ loss. Within the
PR class, Di~Paola's theorem on Church-random sets in
subrecursive hierarchies~\cite{DiPaola1969}, combined
with our aggregation result, yields for every $n\geq3$
a PR sequence on which all rational-valued
$\mathcal E^n$ forecasters have asymptotic loss $1/2$,
while the PR-superpredictor achieves vanishing loss.
The finite-state separation follows as a special case.
Together, these results delineate the computational
limits of prediction and the advantages of universal
aggregation.
\al{
Finally, the framework admits a natural extension to logarithmic
loss through computable Bayesian mixtures, in the spirit of
Vovk~\cite{vovk1992universal}, yielding regret bounded
independently of the sequence length against each fixed
primitive-recursive forecaster. A detailed treatment is left for
future work.
}
}

\section*{Acknowledgment}

The author used ChatGPT by OpenAI and Claude by Anthropic as interactive research and writing
assistants during the development of this work, including for
testing and refining arguments, organizing proofs, and improving the exposition. All mathematical claims, proofs, references, and final
formulations were independently reviewed and verified by the author, who
takes full responsibility for the content of the paper. We would like to thank Prof. Meir Feder for helpful discussions regarding the manuscript, and for suggesting using log-loss as a relevant extension.

\appendices

\section{Exponential-Weights and Dyadic-Window Derivations}
\label{app:exp-weights}

\al{We give the details of the sleeping-expert potential argument underlying
Theorem~\ref{thm:pr-aggregation}. The proof accommodates a countable prior,
expert activation within learning-rate epochs, and a nonincreasing learning
rate. In particular, the activation times $\tau_e=e(e+1)/2$ are distinct
from the dyadic epoch boundaries $b_r=2^r-1$.}

\al{Fix an individual sequence $X$ and write $x_n=X_n$. At round $n$,
let $p_{e,n}$ be expert $e$'s extended advice:}
\begin{equation}
\al{
 p_{e,n}:=
 \begin{cases}
  q_n,&n<\tau_e,\\
  P_e(X_0^{n-1}),&n\geq\tau_e.
 \end{cases}}
\end{equation}
\al{The corresponding extended loss and cumulative loss are}
\begin{equation}
\al{
 \ell_{e,n}:=\ell(p_{e,n},x_n)
   =p_{e,n}(1-x_n)+(1-p_{e,n})x_n,
 \qquad
 C_{e,n}:=\sum_{t<n}\ell_{e,t}.}
\label{eq:appendix-expert-loss}
\end{equation}
\al{Thus $C_{e,n}=\widetilde L_{e,n}$ in the main text. While expert $e$
is sleeping, $C_{e,n}=L_n(Q_{\mathrm{PR}},X)$; in particular, at
activation $C_{e,\tau_e}=L_{\tau_e}(Q_{\mathrm{PR}},X)$. Thereafter the
expert's actual losses are added incrementally. Each extended
per-round loss lies in $[0,1]$.}

\al{For $\eta>0$ introduce the potential and its normalized logarithm}
\begin{equation}
\al{
 Z_n(\eta):=\sum_{e\geq0}\pi_e e^{-\eta C_{e,n}},
 \qquad
 F_n(\eta):=-\frac{1}{\eta}\log Z_n(\eta).}
\label{eq:appendix-potential}
\end{equation}
\al{Because $0\leq C_{e,n}\leq n$ and $\sum_e\pi_e=1$,
$e^{-\eta n}\leq Z_n(\eta)\leq1$. Hence the countable sum is finite
and strictly positive. Define the normalized weights}
\begin{equation}
\al{
 \alpha_{e,n}:=
 \frac{\pi_e e^{-\eta_n C_{e,n}}}{Z_n(\eta_n)},
 \qquad \sum_e\alpha_{e,n}=1.}
\label{eq:appendix-normalized-weights}
\end{equation}

\al{We first verify that the countable potential is compatible with the
finite, computable mixture. With $A_n=\{e:\tau_e\leq n\}$, sleeping
experts have advice $q_n$. Consequently, the extended-advice mixture
satisfies}
\al{\begin{align}
 q_n
 &=\sum_{e\in A_n}\alpha_{e,n}P_e(X_0^{n-1})
   +q_n\sum_{e\notin A_n}\alpha_{e,n} \notag\\
 &\Longleftrightarrow\quad
 q_n=\frac{\sum_{e\in A_n}\pi_e e^{-\eta_n C_{e,n}}
                       P_e(X_0^{n-1})}
              {\sum_{e\in A_n}\pi_e e^{-\eta_n C_{e,n}}}.
\label{eq:appendix-active-cancellation}
\end{align}}
\al{The denominator is positive because $0\in A_n$ for all $n$.
The last expression is precisely~\eqref{eq:adaptive-pr-mixture}, and
only finitely many experts are evaluated at any round. The full
countable mixture is used for analysis, not for computation.}

\al{The online loss state is retained when the learning rate changes.
At a dyadic boundary $b_r$, $r\geq1$, set
$\rho_r:=\eta_r/\eta_{r-1}=2^{-1/2}$. For an expert already active,
its weight at the new rate is obtained from its accumulated loss:}
\begin{equation}
\al{
 \pi_e e^{-\eta_r C_{e,b_r}}
 =\pi_e^{\,1-\rho_r}
  \bigl(\pi_e e^{-\eta_{r-1}C_{e,b_r}}\bigr)^{\rho_r}.}
\label{eq:appendix-retempering}
\end{equation}
\al{Thus a rate change re-tempers the weights; it neither resets
cumulative losses nor changes the expert activation schedule.}

\al{Let $\widehat\ell_n:=\ell(q_n,x_n)$ denote the master's conditional
expected zero--one loss. Since this loss is affine in the forecast,}
\al{\begin{align}
 \widehat\ell_n
 &=\ell\!\left(\sum_e\alpha_{e,n}p_{e,n},x_n\right)
   \notag\\
 &=\sum_e\alpha_{e,n}
       \bigl[p_{e,n}(1-x_n)+(1-p_{e,n})x_n\bigr]
   \notag\\
 &=\sum_e\alpha_{e,n}\ell_{e,n}.
\label{eq:appendix-loss-mixture}
\end{align}}
\al{This is an equality, rather than merely a convexity bound.}

\al{To compare successive potentials, keep the learning rate fixed at
$\eta_n$ during the $n$th update. Directly from the definitions,}
\begin{equation}
\al{
 \frac{Z_{n+1}(\eta_n)}{Z_n(\eta_n)}
   =\sum_e\alpha_{e,n}e^{-\eta_n\ell_{e,n}}.}
\label{eq:appendix-potential-ratio}
\end{equation}
\al{Hoeffding's lemma, applied to the distribution
$\{\alpha_{e,n}\}_{e\geq0}$ and losses in $[0,1]$, gives}
\begin{equation}
\al{
 \log\sum_e\alpha_{e,n}e^{-\eta_n\ell_{e,n}}
 \leq-\eta_n\sum_e\alpha_{e,n}\ell_{e,n}
      +\frac{\eta_n^2}{8}.}
\label{eq:appendix-hoeffding}
\end{equation}
\al{Define the mix loss $m_n$ by}
\begin{equation}
\al{
 m_n:=-\frac{1}{\eta_n}
       \log\frac{Z_{n+1}(\eta_n)}{Z_n(\eta_n)}
     =F_{n+1}(\eta_n)-F_n(\eta_n).}
\label{eq:appendix-mix-loss}
\end{equation}
\al{Combining~\eqref{eq:appendix-loss-mixture}--
\eqref{eq:appendix-mix-loss} yields the one-round estimate}
\begin{equation}
\al{\widehat\ell_n\leq m_n+\frac{\eta_n}{8}.}
\label{eq:appendix-one-step-bound}
\end{equation}

\al{The remaining issue is telescoping with a changing learning rate.
For $0<\eta'\leq\eta$, put $a=\eta'/\eta\in(0,1]$.
By concavity of $u\mapsto u^a$ and Jensen's inequality,}
\al{\begin{align}
 Z_n(\eta')
 &=\sum_e\pi_e\bigl(e^{-\eta C_{e,n}}\bigr)^a
   \notag\\
 &\leq\left(\sum_e\pi_e e^{-\eta C_{e,n}}\right)^a
   = Z_n(\eta)^a.\end{align}}
\al{Taking logarithms and multiplying by $-1/\eta'$ gives}
\begin{equation}
\al{F_n(\eta')\geq F_n(\eta)
 \qquad(0<\eta'\leq\eta).}
\label{eq:appendix-potential-monotonicity}
\end{equation}
\al{Since $\eta_n$ is nonincreasing, we can now exhibit the boundary
terms in the telescoping sum explicitly:}
\al{\begin{align}
 \sum_{n=0}^{N-1}m_n
 &=F_N(\eta_{N-1})-F_0(\eta_0)
   \notag\\
 &\quad+\sum_{n=1}^{N-1}
      \bigl[F_n(\eta_{n-1})-F_n(\eta_n)\bigr]
   \notag\\
 &\leq F_N(\eta_{N-1}).
\label{eq:appendix-telescoping}
\end{align}}
\al{Here each term in the last sum is nonpositive by
\eqref{eq:appendix-potential-monotonicity}, and
$F_0(\eta_0)=0$ since $C_{e,0}=0$ and $\sum_e\pi_e=1$.
Summing the one-round bounds therefore gives}
\begin{equation}
\al{
 L_N(Q_{\mathrm{PR}},X)
 \leq F_N(\eta_{N-1})+\frac18\sum_{n<N}\eta_n.}
\label{eq:appendix-master-potential-bound}
\end{equation}

\al{For any comparator $e$, the $e$th term of the potential yields}
\begin{equation}
\al{
 Z_N(\eta_{N-1})
 \geq\pi_e e^{-\eta_{N-1}C_{e,N}},
 \qquad
 F_N(\eta_{N-1})
 \leq C_{e,N}+\frac{\log(1/\pi_e)}{\eta_{N-1}}.}
\label{eq:appendix-comparator-potential}
\end{equation}
\al{Together with~\eqref{eq:appendix-master-potential-bound}, this
proves the extended-loss bound}
\begin{equation}
\al{
 L_N(Q_{\mathrm{PR}},X)-\widetilde L_{e,N}
 \leq\frac{\log(1/\pi_e)}{\eta_{N-1}}
       +\frac18\sum_{n<N}\eta_n.}
\label{eq:appendix-extended-regret}
\end{equation}

\al{For $N>\tau_e$, separating the sleeping and active portions gives}
\begin{equation}
\al{
 \widetilde L_{e,N}
 =L_{\tau_e}(Q_{\mathrm{PR}},X)
   +L_N(P_e,X)-L_{\tau_e}(P_e,X).}
\label{eq:appendix-sleeping-decomposition}
\end{equation}
\al{As every per-round loss belongs to $[0,1]$,}
\begin{equation}
\al{
 \widetilde L_{e,N}-L_N(P_e,X)
 =L_{\tau_e}(Q_{\mathrm{PR}},X)
  -L_{\tau_e}(P_e,X)
 \leq\tau_e.}
\label{eq:appendix-activation-penalty}
\end{equation}
\al{Substitution into~\eqref{eq:appendix-extended-regret}
proves~\eqref{eq:adaptive-pr-regret}, including the explicit
activation penalty.}

\al{Finally, choose $R\geq0$ such that $2^R\leq N<2^{R+1}$.
Because $I_r=\{2^r-1,\ldots,2^{r+1}-2\}$, we have
$N-1\in I_R$ and hence}
\begin{equation}
\al{\eta_{N-1}^{-1}=2^{R/2}\leq\sqrt N.}
\end{equation}
\al{Every full epoch $I_r$ has $2^r$ rounds. Bounding the last,
possibly incomplete epoch by a complete one gives}
\al{\begin{align}
 \sum_{n<N}\eta_n
 &\leq\sum_{r=0}^{R}2^r2^{-r/2}
   =\sum_{r=0}^{R}2^{r/2}
   \notag\\
 &\leq\frac{\sqrt2}{\sqrt2-1}\,2^{R/2}
   \leq(2+\sqrt2)\sqrt N.
\label{eq:appendix-dyadic-rate-sum}
\end{align}}
\al{Using $\tau_e=e(e+1)/2$ and
$\log(1/\pi_e)=\log((e+1)(e+2))$ in
\eqref{eq:adaptive-pr-regret}, we obtain the claimed explicit bound}
\begin{equation}
\al{
 L_N(Q_{\mathrm{PR}},X)-L_N(P_e,X)
 \leq\frac{e(e+1)}2
   +\sqrt N\left[
      \log((e+1)(e+2))+\frac{2+\sqrt2}{8}
    \right],}
\label{eq:appendix-final-regret}
\end{equation}
\al{which is~\eqref{eq:pr-mixture-regret}.}

\section{Effective Bayes Lower Bound}
\label{app:effective-bayes-lower-bound}

This appendix contains the complete proof of Lemma~\ref{lem:probabilistic-bayes-lower-bound}.

\begin{proof}[Proof of Lemma~\ref{lem:probabilistic-bayes-lower-bound}]
Let
\(
    \mathbf X=(\mathbf X_n)_{n\geq0}
\)
denote the canonical coordinate process on $2^\omega$ under $\mu$, and
let
\(
    \mathcal F_n
    :=
    \sigma(\mathbf X_0,\ldots,\mathbf X_{n-1}).
\)
Define
\[
    q_n:=Q(\mathbf X_0^{n-1}),
    \qquad
    Z_n:=\ell(q_n,\mathbf X_n).
\]
Thus $Z_n$ is the realized soft loss of $Q$ at time $n$, and
$0\leq Z_n\leq1$.

For a finite string $\sigma\in2^n$ with $\mu[\sigma]>0$, define
\[
    p(\sigma)
    :=
    \mu(\mathbf X_n=1\mid\mathbf X_0^{n-1}=\sigma)
    =
    \frac{\mu[\sigma1]}{\mu[\sigma]}
\quad \hbox{and} \quad 
    {\bar\ell}(\sigma)
    :=
    p(\sigma)(1-Q(\sigma))
    +(1-p(\sigma))Q(\sigma).
\]
Choose the following version of the conditional expected loss:
\[
    {\bar\ell}_n
    :=
    \mathbb E_\mu[Z_n\mid\mathcal F_n].
\]
On every positive-measure cylinder $[\sigma]$ of length $n$, set
\[
    {\bar\ell}_n(\mathbf x)={\bar\ell}(\sigma),
    \qquad
    \mathbf x\in[\sigma],
\]
and define ${\bar\ell}_n$ arbitrarily on cylinders of $\mu$-measure zero.  This
does not affect its status as a version of the conditional expectation.

Set
\[
    D_n:=Z_n-{\bar\ell}_n,
    \qquad
    S_N:=\sum_{n<N}D_n.
\]
Then $D_n$ is $\mathcal F_{n+1}$-measurable and
\[
    \mathbb E_\mu[D_n\mid\mathcal F_n]=0.
\]
Hence $(S_N,\mathcal F_N)_{N\geq0}$ is a martingale.  Moreover,
$|D_n|\leq1$.  If $m<n$, then $D_m$ is $\mathcal F_n$-measurable, and
therefore
\[
\begin{aligned}
    \mathbb E_\mu[D_mD_n]
    &=
    \mathbb E_\mu\!\left[
       D_m\,\mathbb E_\mu[D_n\mid\mathcal F_n]
    \right]\\
    &=0.
\end{aligned}
\]
Thus the martingale differences are orthogonal in $L^2(\mu)$, and
\[
\begin{aligned}
    \mathbb E_\mu[S_N^2]
    &=
    \sum_{n<N}\mathbb E_\mu[D_n^2]\\
    &\leq N.
\end{aligned}
\]
In particular,
\[
    \mathbb E_\mu[S_N]=0,
    \qquad
    \operatorname{Var}_\mu(S_N)\leq N.
\]

\al{
The connection between martingale convergence,
information divergence, and Shannon--McMillan--Breiman-type
results was studied by Harremo\"es
\cite{HarremoesMartingales2005}.
In the setting of algorithmic randomness,
Steifer~\cite[Lem.~161]{Steifer2019PhD}
established effective martingale convergence for
computable deterministic predictors on Martin-L\"of
random sequences. We prove the corresponding statement
for the bounded soft losses of computable probabilistic
forecasters, using a second-moment estimate along the
square subsequence in place of Azuma's inequality:
}
 
For a string
\[
    \sigma=\sigma_0\cdots\sigma_{N-1}\in2^N
\]
with $\mu[\sigma]>0$, define
\[
    s_N(\sigma)
    :=
    \sum_{n<N}
    \left[
       \ell\bigl(Q(\sigma_0^{n-1}),\sigma_n\bigr)
       -{\bar \ell}(\sigma_0^{n-1})
    \right].
\]
Since $\mu[\sigma]>0$, every prefix of $\sigma$ also has positive
$\mu$-measure, so all the terms in this expression are well defined.

The random variable $S_N$ is $\mathcal F_N$-measurable and therefore
constant on every cylinder of length $N$.  More precisely, whenever
$\mu[\sigma]>0$,
\begin{equation}
    S_N(\mathbf x)=s_N(\sigma)
    \qquad
    \text{for every }\mathbf x\in[\sigma].
    \label{eq:martingale-cylinder-value}
\end{equation}

Fix an integer $j\geq1$.  For each $k\geq1$, define
\[
    U_{k,j}
    :=
    \bigcup
    \left\{
       [\sigma]:
       \sigma\in2^{k^2},\
       \mu[\sigma]>0,\
       |s_{k^2}(\sigma)|>\frac{k^2}{j}
    \right\}.
\]
We claim that $(U_{k,j})_{k\geq1}$ is uniformly effectively open.  Since
$\mu$ is computable, the condition $\mu[\sigma]>0$ is semidecidable:
one searches for a positive rational lower bound for $\mu[\sigma]$.
Once such a bound has been found, every ratio
\[
    \frac{\mu[\tau1]}{\mu[\tau]},
    \qquad
    \tau\preceq\sigma,
\]
is computable uniformly from $\mu$ and $\sigma$.  Since $Q$ is total 
computable, it follows that $s_{k^2}(\sigma)$ is a computable real,
uniformly in $k$ and $\sigma$.  The defining inequality is strict, so
\[
    |s_{k^2}(\sigma)|>\frac{k^2}{j}
\]
is semidecidable.  Hence the cylinders forming $U_{k,j}$ can be
enumerated effectively.

By~\eqref{eq:martingale-cylinder-value}, $U_{k,j}$ agrees, outside a
finite union of zero-measure cylinders, with the martingale-deviation
event
\[
    \left\{
       \mathbf x:
       |S_{k^2}(\mathbf x)|>\frac{k^2}{j}
    \right\}.
\]
Consequently,
\[
    \mu(U_{k,j})
    =
    \mu\left\{
       |S_{k^2}|>\frac{k^2}{j}
    \right\}.
\]
Since $\mathbb E_\mu[S_{k^2}]=0$ and
\[
    \operatorname{Var}_\mu(S_{k^2})
    \leq k^2,
\]
Chebyshev's inequality gives
\[
\begin{aligned}
    \mu(U_{k,j})
    &\leq
    \frac{
       \operatorname{Var}_\mu(S_{k^2})
    }{
       k^4/j^2
    }\\
    &\leq
    \frac{j^2}{k^2}.
\end{aligned}
\]
The bounds have a computable summable tail, since, for $K\geq2$,
\[
    \sum_{k=K}^{\infty}\frac{j^2}{k^2}
    \leq
    \frac{j^2}{K-1}.
\]
The effective Borel-Cantelli lemma
\cite{DebowskiSteifer2022} therefore implies that every
$X\in\operatorname{MLR}_\mu$ belongs to only finitely many of the sets
$U_{k,j}$.

Every prefix of a $\mu$-Martin-L\"of random sequence has positive
$\mu$-measure.  Hence, for the fixed sequence $X$,
\[
    X\in U_{k,j}
    \quad\Longleftrightarrow\quad
    |S_{k^2}(X)|>\frac{k^2}{j}.
\]
It follows that, for every fixed $j\geq1$,
\[
    \frac{|S_{k^2}(X)|}{k^2}\leq\frac1j
\]
for all sufficiently large $k$.  Since $j$ is arbitrary,
\[
    \frac{S_{k^2}(X)}{k^2}\longrightarrow0.
\]

To pass from the quadratic subsequence to all times, let
\[
    k^2\leq N<(k+1)^2.
\]
Since $|D_n|\leq1$,
\[
    |S_N-S_{k^2}|
    \leq N-k^2
    \leq2k.
\]
Therefore
\[
    \frac{|S_N(X)|}{N}
    \leq
    \frac{|S_{k^2}(X)|}{k^2}
    +\frac{2k}{k^2}
    \longrightarrow0.
\]
We have proved that
\begin{equation}
    \frac1N
    \sum_{n<N}
    \bigl(Z_n(X)-{\bar\ell}_n(X)\bigr)
    \longrightarrow0.
    \label{eq:martingale-loss-second-moment}
\end{equation}

For the fixed sequence $X\in\operatorname{MLR}_\mu$, put
\[
    p_n
    :=
    \mu(X_n=1\mid X_0^{n-1})
    =
    \frac{\mu[X_0^{n-1}1]}
         {\mu[X_0^{n-1}]},
    \qquad
    q_n:=Q(X_0^{n-1}).
\]
Then
\[
    {\bar\ell}_n(X)
    =
    p_n(1-q_n)+(1-p_n)q_n.
\]
For every $p,q\in[0,1]$,
\[
    p(1-q)+(1-p)q
    \geq
    \min\{p,1-p\},
\]
because randomization cannot improve upon the deterministic Bayes
action for a fixed binary conditional distribution.  Hence
\begin{equation}
    {\bar\ell}_n(X)\geq\min\{p_n,1-p_n\}.
    \label{eq:soft-loss-bayes-lower}
\end{equation}

The effective L\'evy law and the effective Breiman ergodic theorem,
applied as in the proof of the effective source-prediction theorem
\cite{DebowskiSteifer2022}, imply that
\begin{equation}
    \frac1N
    \sum_{n<N}\min\{p_n,1-p_n\}
    \longrightarrow e_\mu
    \label{eq:finite-past-bayes-limit}
\end{equation}
on every $X\in\operatorname{MLR}_\mu$.

Finally, by
\eqref{eq:martingale-loss-second-moment},
\eqref{eq:soft-loss-bayes-lower}, and
\eqref{eq:finite-past-bayes-limit},
\[
\begin{aligned}
    \frac{L_N(Q,X)}{N}
    =
    \frac1N\sum_{n<N}Z_n(X)
    =
    \frac1N\sum_{n<N}{\bar\ell}_n(X)+o(1)
    \geq
    \frac1N
    \sum_{n<N}\min\{p_n,1-p_n\}
    +o(1).
\end{aligned}
\]
Therefore
\[
    \liminf_{N\to\infty}
       \frac{L_N(Q,X)}{N}
    \geq e_\mu.
\]
\al{More generally, for any oracle $A$, the same
conclusion holds when $Q$ is total and
$A$-computable and $X\in\MLR_\mu^A$.}

\end{proof}

\section{Proof of the Primitive-Recursive Switching Theorem}
\label{app:pr-switching-proof}

This appendix contains the complete proof deferred from the main text.

\begin{proof}[Proof of Theorem~\ref{thm:pr-switched-ergodic}]
For a finite string $\sigma$ of length $n$ and $j$, let
\[
    \sigma^{[j]}
       :=
    \bigl\langle \sigma_t:t<n,\ s(t)=j\bigr\rangle
\]
be the subsequence previously observed while state $j$ was active.  Since
$s$ is primitive recursive, the map
\[
    (j,\sigma)\longmapsto\sigma^{[j]}
\]
is primitive recursive: it is obtained by bounded search through the
positions $t<|\sigma|$.

Define the state-aware predictor
\begin{equation}
    P_s(\sigma)
       :=
    P_{\mathrm{PPM}}
       \bigl(\sigma^{[s(|\sigma|)]}\bigr).
    \label{eq:state-aware-ppm}
\end{equation}
By Lemma~\ref{lem:ppm-pr}, $P_{\mathrm{PPM}}$ is primitive recursive.
Hence $P_s\in\mathcal P_{\mathrm{PR}}$.

For every $N$, the loss of $P_s$ decomposes exactly as
\begin{equation}
    L_N(P_s,X)
       =
    \sum_{j}
       L_{N_j(N)}
          \bigl(P_{\mathrm{PPM}},Y^{(j)}\bigr).
    \label{eq:statewise-loss-decomposition}
\end{equation}
A computable projection of a Martin-L\"of random point is Martin-L\"of
random for the corresponding pushforward measure.  Therefore each
$Y^{(j)}$ is Martin-L\"of random for $\mu_j$.  PPM universality gives
\[
    L_m(P_{\mathrm{PPM}},Y^{(j)})
       =
    me_j+o(m)
\]
along every component whose visit count tends to infinity.  Components
visited only finitely often contribute $O(1)$.  Since the number of
states is finite, \eqref{eq:statewise-loss-decomposition} implies
\begin{equation}
    L_N(P_s,X)
       =
    \sum_{j}N_j(N)e_j+o(N).
    \label{eq:state-aware-upper}
\end{equation}
Theorem~\ref{thm:pr-aggregation}, applied to the fixed
primitive-recursive expert $P_s$, now yields
\begin{equation}
    L_N(Q_{\mathrm{PR}},X)
       \leq
    \sum_{j}N_j(N)e_j+o(N).
    \label{eq:aggregate-switched-upper}
\end{equation}

It remains to prove the matching lower bound.  Fix a computable
probabilistic predictor $Q$.  For each state $j$ that occurs infinitely
often, let
\[
    n_{j,0}<n_{j,1}<\cdots
\]
be the computable increasing enumeration of the times at which
$s(n)=j$, and let
\[
    Z^{(j)}:=\bigoplus_{i\neq j}Y^{(i)}
\]
code all component streams other than the $j$-th one.  From an oracle for
$Z^{(j)}$, a prefix $Y^{(j)}_0,\ldots,Y^{(j)}_{k-1}$, and the schedule $s$,
one can reconstruct the complete interleaved prefix $X_0^{n_{j,k}-1}$.
Consequently,
\[
    Q_j^{Z^{(j)}}((Y^{(j)})_0^{k-1})
       :=
    Q(X_0^{n_{j,k}-1})
\]
defines a probabilistic predictor that is computable relative to
$Z^{(j)}$.

Let
\[
    \nu_j:=\bigotimes_{i\neq j}\mu_i,
\]
where the product is viewed as a measure on the join
\[
    Z^{(j)}=\bigoplus_{i\neq j}Y^{(i)}.
\]
A computable permutation and grouping of the coordinates transforms
$\boldsymbol{\mu}$ into
\[
    \nu_j\otimes\mu_j
\]
on the ordered pair $(Z^{(j)},Y^{(j)})$.  Since computable
measure-preserving coordinate permutations preserve Martin-L\"of
randomness, the pair
\[
    (Z^{(j)},Y^{(j)})
\]
is Martin-L\"of random with respect to
$\nu_j\otimes\mu_j$.  The generalized van Lambalgen
theorem \cite{Takahashi2011} therefore implies
\[
    Y^{(j)}
    \in\operatorname{MLR}_{\mu_j}^{Z^{(j)}}.
\]
The induced predictor $Q_j^{Z^{(j)}}$ is computable relative to the same
oracle $Z^{(j)}$.  Hence the relativized Bayes lower bound gives
\begin{equation}
    \liminf_{m\to\infty}
    \frac{
       L_m(Q_j^{Z^{(j)}},Y^{(j)})-m e_j
    }{m}
    \geq 0.
    \label{eq:relativized-component-lower}
\end{equation}
For a state visited only finitely often, its total contribution is $O(1)$.
The loss of $Q$ on the interleaved sequence decomposes exactly into its
losses on the active component rounds:
\[
    L_N(Q,X)
       =
    \sum_{j}
       L_{N_j(N)}
          (Q_j^{Z^{(j)}},Y^{(j)}).
\]
Since there are only finitely many states, summing
\eqref{eq:relativized-component-lower} at $m=N_j(N)$ yields
\[
    L_N(Q,X)
       \geq
    \sum_{j}N_j(N)e_j-o(N).
\]
This proves \eqref{eq:pr-switched-lower}.  Applying it to
$Q_{\mathrm{PR}}$ and combining it with
\eqref{eq:aggregate-switched-upper} proves
\eqref{eq:pr-switched-optimality}.  The statement with limiting state
frequencies follows after division by $N$.
\end{proof}

\end{document}